\documentclass{IEEEtran}
\usepackage[utf8]{inputenc}

\usepackage{amsmath,amssymb,amsfonts,amsthm,commath,MnSymbol}
\usepackage{mathtools}
\usepackage{algorithm}
\usepackage[noend]{algpseudocode}
\usepackage{hhline}
\usepackage[caption=false]{subfig}
\usepackage{multirow}
\usepackage{siunitx}
\usepackage{graphicx}
\usepackage{textcomp}
\usepackage{xcolor}
\usepackage[sorting=ynt, style=ieee]{biblatex}
\usepackage{booktabs}
\usepackage{hyperref}
\usepackage{marginnote}

\usepackage{scalerel} 
\usepackage{tikz} 
\usetikzlibrary{svg.path} \definecolor{orcidlogocol}{HTML}{A6CE39} \tikzset{ orcidlogo/.pic={ \fill[orcidlogocol] svg{M256,128c0,70.7-57.3,128-128,128C57.3,256,0,198.7,0,128C0,57.3,57.3,0,128,0C198.7,0,256,57.3,256,128z}; \fill[white] svg{M86.3,186.2H70.9V79.1h15.4v48.4V186.2z} svg{M108.9,79.1h41.6c39.6,0,57,28.3,57,53.6c0,27.5-21.5,53.6-56.8,53.6h-41.8V79.1z M124.3,172.4h24.5c34.9,0,42.9-26.5,42.9-39.7c0-21.5-13.7-39.7-43.7-39.7h-23.7V172.4z} svg{M88.7,56.8c0,5.5-4.5,10.1-10.1,10.1c-5.6,0-10.1-4.6-10.1-10.1c0-5.6,4.5-10.1,10.1-10.1C84.2,46.7,88.7,51.3,88.7,56.8z}; } } 
\newcommand\orcid[1]{\href{https://orcid.org/#1}{\mbox{\scalerel*{
\begin{tikzpicture}[yscale=-1,transform shape]
\pic{orcidlogo};
\end{tikzpicture}
}{|}}}}
\newtheorem{thm}{Theorem}
\newtheorem{prop}[thm]{Proposition}

\newtheorem{definition}[thm]{Definition}

\newcommand{\algmargin}{\the\ALG@thistlm}
\makeatother
\newlength{\whilewidth}
\algdef{SE}[parWHILE]{parWhile}{EndparWhile}[1]
  {\parbox[t]{\dimexpr\linewidth-\algmargin}{%
     \hangindent\whilewidth\strut\algorithmicwhile\ #1\ \algorithmicdo\strut}}{\algorithmicend\ \algorithmicwhile}%
\algnewcommand{\LineComment}[1]{\State \(\triangleright\) #1}
\algnewcommand{\parState}[1]{\State%
  \parbox[t]{\dimexpr\linewidth-\algmargin}{\strut #1\strut}}

\providecommand{\keywords}[1]{\textbf{\textit{Index terms---}} #1}

\usepackage{marginnote}

\newcommand{\changedOne}[1]{{#1}}

\begin{document}

\title{Structural Equivalence in Subgraph Matching}
\author{Dominic Yang \orcid{0000-0002-9453-2299}, Yurun Ge \orcid{0000-0001-9666-9988
}, Thien Nguyen \orcid{0000-0002-5972-7189}, Denali Molitor \orcid{0000-0001-9750-3533}, Jacob D. Moorman \orcid{0000-0002-4291-1561}, Andrea L. Bertozzi \orcid{0000-0003-0396-7391}, {\it Member, IEEE}.\IEEEcompsocitemizethanks{\IEEEcompsocthanksitem Dominic Yang, Yurun Ge, Denali Molitor, Jacob Moorman, and Andrea L. Bertozzi are with the Department of Mathematics, University of California, Los Angeles, Los Angeles, CA, 90095. E-mail: domyang@math.ucla.edu, yurun@math.ucla.edu, dmolitor@math.ucla.edu, jdmoorman@math.ucla.edu, bertozzi@math.ucla.edu. \IEEEcompsocthanksitem Thien Nguyen is with the department of Computer Science at Northeastern University, Boston, MA, 02115. E-mail: nguyen.thien@northeastern.edu}}
\date{\today}
\maketitle

\begin{abstract}
Symmetry plays a major role in subgraph matching both in the description of the graphs in question and in how it confounds the search process. This work addresses how to quantify these effects and how to use symmetries to increase the efficiency of subgraph isomorphism algorithms. 
We introduce rigorous definitions of structural equivalence and establish conditions for when it can be safely used to generate more solutions.
We illustrate how to adapt standard search routines to utilize these symmetries to accelerate search and compactly describe the solution space.
We then adapt a state-of-the-art solver and perform a comprehensive series of tests to demonstrate these methods' efficacy on a standard benchmark set. 
We extend these methods to multiplex graphs and present results on large multiplex networks drawn from transportation systems, social media, adversarial attacks, and knowledge graphs.
\end{abstract}
\keywords{Subgraph isomorphism, subgraph matching, multiplex network, structural equivalence, graph structure}
\section{Introduction}

The subgraph isomorphism problem (also called the subgraph matching problem) specifies a small graph (the \textbf{template})
to find as a subgraph within a larger (\textbf{world}) graph.
This problem has been well-studied especially in the pattern recognition community. 
The surveys \cite{foggia2014graph}, \cite{Conte2007}, and \cite{emmert2016fifty} explain the broad variety of techniques used as well as applications including handwriting recognition \cite{sanfeliu1983distance}, face recognition \cite{wiskott1997}, biomedical uses \cite{dumay1992consistent}, sudoku puzzles and adversarial activity \cite{moorman2021subgraph}. 
More recently, subgraph matching arises as a component in motif discovery \cite{micale2018fast, ribeiro2014discovering}, where frequent subgraphs are uncovered for graph analysis in domains including social networks and biochemical data. 
Additionally, subgraph matching is relevant in knowledge graph searches, wherein incomplete factual statements are completed by querying a knowledge database \cite{auer2007dbpedia, inexact20}. 

Networks are present in many applications; hence, the ability to detect interesting structures, i.e., subgraphs, apparent in the networks bears great importance. We investigate subgraph matching on a wide variety of networks, simulated and real, single channel and multichannel, ranging from hundreds to millions of nodes. These data sets include biochemical reactions \cite{gay2014subgraph}, pattern recognition \cite{damiand2011polynomial}, transportation networks \cite{fan2012graph}, social networks \cite{de2013anatomy}, and knowledge graphs \cite{zucker2021leveraging}.

This paper addresses exact subgraph isomorphisms: given a template $G_T = (V_T, E_T)$, and a world $G_W = (V_W, E_W)$, find a mapping $f: V_T \rightarrow V_W$ that is both injective and respects the structure of $G_T$. For the latter property to hold, we require that if $(t_1, t_2) \in E_T$, then we must have $(f(t_1), f(t_2)) \in E_W$. If this is true, we say that $f$ is \textbf{edge-preserving}. We define subgraph isomorphism as follows:

\begin{definition}
    Given a template $G_T = (V_T, E_T)$ and a world $G_W = (V_W, E_W)$, a map $f:V_T \rightarrow V_W$ is a \textbf{subgraph isomorphism} if and only if $f$ is injective and edge-preserving.
\end{definition}

%Often it is the case that the graphs we study have additional structure in the form of multiple edges or labels on vertices and edges. To handle multiple edges, we require that the number of edges between $f(u)$ and $f(v)$ exceed that of $u$ and $v$. If there is a vertex-labeling function $L_V$, we additionally require that a subgraph isomorphism be \textbf{label-preserving}: $L_V(f(v)) = L_V(v)$ and $L_E(f(e)) = L_E(e)$ for all $v \in V_T, e \in E_T$ where $L_V: V\rightarrow\mathcal{L}_V, L_E: E\rightarrow\mathcal{L}_E$ and $\mathcal{L}_V$ and $\mathcal{L}_E$ are the spaces of labels for vertices and edges, respectively.

Throughout this paper, we use the terms subgraph isomorphism and subgraph matching interchangeably.
Related terms are \textbf{subgraph homomorphism} which relaxes the injectivity requirement, and \textbf{induced subgraph isomorphism} which also requires the map to be non-edge-preserving (if $(u,v) \notin E_T, (f(u),f(v)) \notin E_W$). 
%For clarity of exposition, unless otherwise specified, we consider simple directed graphs in this paper; however, the concepts discussed easily extend to multigraphs with both vertex and edge labels.

We are interested in the subgraph matching problem (SMP) \cite{moorman2021subgraph}:
\begin{definition}[\textbf{Subgraph Matching Problem}]
    Given a template graph $G_T$ and a world graph $G_W$, find all subgraph isomorphisms from $G_T$ to $G_W$.
\end{definition}
If there is at least one subgraph isomorphism, we call the problem \textbf{satisfiable}. 

\begin{figure}
    \centering
    \includegraphics[width=\linewidth]{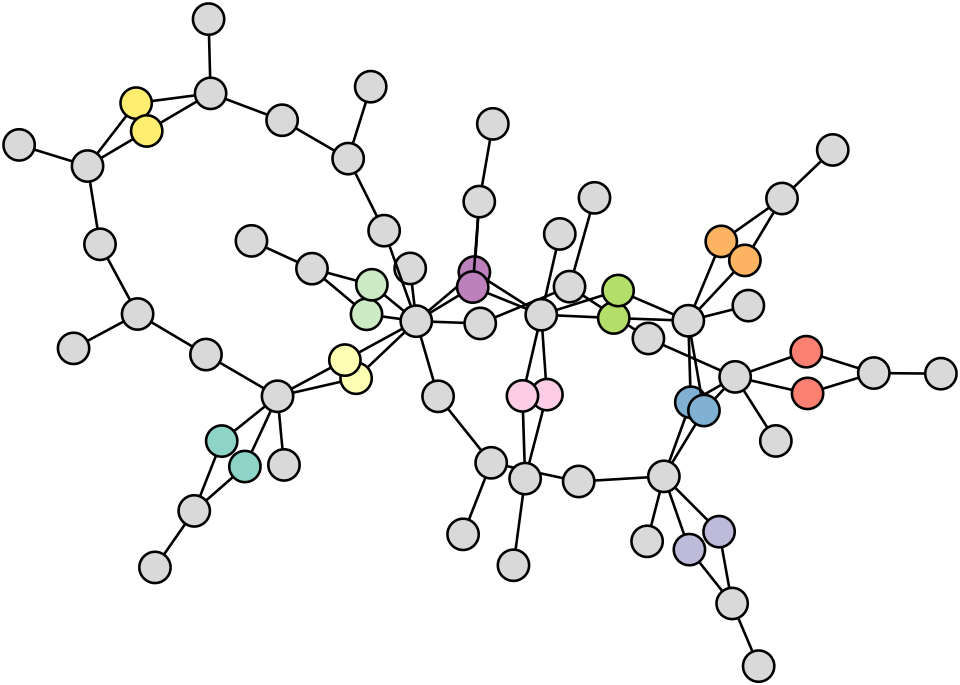}
    \caption{Graph representing a system of biochemical reactions from \cite{gay2014subgraph}. Non-gray nodes of the same color are structurally equivalent.}
    \label{fig:biochemical-uncompressed}
\end{figure}

%We also introduce some other terms from graph theory which will be referenced in this paper. %We say $H=\left(V_H, E_H\right)$ is a \textit{subgraph} of $G$ if $V_H \subseteq V_G$ and $E_H \subseteq E_G$.
%An \textbf{independent set} of vertices is a subset of $V(G)$ where there are no edges between any two nodes in the subset. In other words, it is a completely disconnected subgraph. 
%A \textbf{clique} is a set $C \subseteq V_G$ of nodes such that for all $u \neq v \in C$, $(u,v)$ and $(v,u)$ are both in $E_G$. 
%Finally, we define the set of \textbf{in-neighbors} of a vertex $v$ by $N_{in}(v) = \{u: (u,v) \in E_G\}$ and the set of \textbf{out-neighbors} by $N_{out}(v) = \{u: (v,u) \in E_G\}$. The set of neighbors, $N(v)$ is just the union of in-neighbors and out-neighbors $N(v) = N_{out}$.

Simply finding a subgraph isomorphism is %easily shown to be
NP-complete \cite{garey2002computers}, suggesting that there is no algorithm that efficiently finds all subgraph isomorphisms on all graphs. In spite of this, significant progress has been made in the development of algorithms for detecting subgraph isomorphisms \cite{mccreesh2020glasgow,Han2013TurboisoTU,bi2016efficient,houbraken2014index}. As in other NP-complete problems, the literature addressing the enumeration of exact subgraph isomorphisms has focused on performing a full tree search of the solution space. The state-of-the-art algorithms generally focus on iteratively building partial matches and use heuristics for optimizing the variable ordering and pruning branches of the search tree.

We are interested in the full enumeration and characterization of the solution space for the SMP. This is important for real-world applications. Consider a template representing transactions of a crime ring and the world is the broader transaction network. If there are thousands of matches, any given match is likely to be a false positive suggesting the need to further narrow down the search by adding more detail to the template.
Our work identifies and characterizes the structure of such redundancies due to symmetry in the matching problem. \changedOne{We exploit these symmetries to produce a compressed version of the solution space which saves space as well as aids understanding the problem's solutions.}

%In this paper, we exploit both static and dynamic forms of symmetry to better understand large solution spaces and compactify them for ease of expression and understanding. 
As an example of symmetry, observe the template graph in \figref{fig:biochemical-uncompressed} which is from a system of biochemical reactions \cite{gay2014subgraph} and note that each pair of colored nodes is interchangeable in any solution as they have the exact same neighbors. 
As there are 11 such pairs in the graph, for any found isomorphism, we can generate $2^{11} = 2048$ more solutions simply by interchanging nodes. 
By avoiding redundant solutions in a subgraph search, we can significantly reduce the search time (potentially by a factor of 2048 or more). 
This simple form of symmetry is known as \textbf{structural equivalence}.

Broader notions of equivalence can be used to further accelerate search. 
In \figref{fig:toy-example}, the yellow and blue nodes are each individually structurally equivalent. 
However, if we proceed by matching A to 1, then we can complete an isomorphism by matching B and C to any of 2, 3, 4, or 5. 
The presence of additional edges incident to 4 and 5 hides that 4 and 5 may be swapped out for 2 or 3. 
By identifying when these additional edges may be ignored, we can again dramatically reduce the amount of work. 
This second notion of equivalence we will refer to as \textbf{candidate equivalence}. In the body of this paper, we will formally define these terms and demonstrate how they can be applied in a tree search algorithm.
These two notions of equivalence can be broadly classified into two categories: {\bf static equivalenc}e, which describes equivalence apparent from the problem description, and {\bf dynamic equivalence}, which describes equivalence uncovered in the search process. Structural equivalence falls into the former category while candidate equivalence belongs to the latter. \changedOne{We note that these forms of equivalence are dependent on the structure of the graphs and cannot be used should we interchange one template graph for another.}

\begin{figure}
    \centering
    \includegraphics[height=100px]{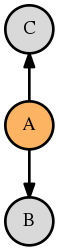} \;\;
    \includegraphics[height=100px]{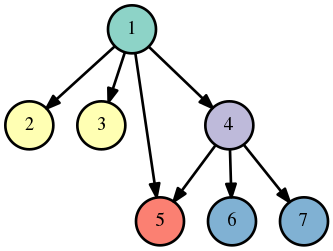}
    \caption{Example subgraph isomorphism problem with template on the left and world on the right. Nodes of the same color are structurally equivalent.}
    \label{fig:toy-example}
\end{figure}

\subsection{Related Work}

The first significant subgraph isomorphism algorithm proposed by Ullmann \cite{ullmann1976algorithm} %did not perform any optimizations on the match ordering, and
generates candidate vertices from the neighbors of already matched world vertices. The widely used VF2 algorithm \cite{cordella2004sub} improves on this by choosing a match ordering that favors template vertices adjacent to already-matched template vertices and adding pruning rules based on the degrees of vertices. The authors more recently published the VF3 algorithm \cite{carletti2017introducing} that further extends this approach.

In a different approach, Solnon \cite{solnon2010alldifferent} emphasizes the  constraint propagation paradigm from artificial intelligence, and her algorithm LAD internally stores candidate lists for every template vertex. By way of a repeated application of a filter based on the vertex neighborhood structure, her algorithm can effectively prune branches of the tree search. She expands on this work in \cite{kothoff2016} to incorporate more powerful filters. Other solvers including Glasgow \cite{mccreesh2015parallel, mccreesh2020glasgow}, SND \cite{audemard2014scoring}, and ILF \cite{zampelli2010solving} each address various ways to enhance filters based on other graph properties including number of paths, cliques or more complicated neighborhoods in order to strengthen the filters.

Significant work has been done to exploit symmetry to compress graphs and count isomorphisms. TurboIso \cite{Han2013TurboisoTU} exploits basic symmetry in the template graph and optimizes the matching order based on a selection of candidate regions and exploration within those regions. CFL-Match \cite{bi2016efficient} proposes a match ordering based on a decomposition of the template graph into the core, a highly connected subgraph, and a forest which is further decomposed into a forest and leaves. BoostIso \cite{ren2015exploiting} exploits symmetry in the world graph and presents a method by which other tree-search-based approaches are accelerated by using their methodology. The ISMA algorithm \cite{demeyer2013index} exploits basic bilateral and rotational symmetry of the template to boost subgraph search and this work was extended into the ISMAGS algorithm \cite{houbraken2014index} to incorporate general automorphic symmetries of the template graph.

\changedOne{One closely related problem is the inexact subgraph matching problem. This problem replaces the strict edge-preserving constraints of exact isomorphism with a penalty for edge and label mismatches which is to be minimized. The exact form of the penalty varies across applications and there are a myriad of approaches many taking inspiration from the exact matching problem. These involve a variety of techniques including filtering \cite{inexact20,InexactHRL2019}, A* search \cite{jin2019noisy}, indexing \cite{MLIndex}, and continuous techniques \cite{MatchedFilter2019}. We will not be considering this problem in this paper.}

\subsection{Paper Outline}
In this paper, we demonstrate how tree search-oriented approaches can be accelerated by exploiting both static forms of equivalence, apparent at the start of search, as well as dynamic forms of equivalence, which are uncovered as the search proceeds.
In Section \ref{sec:structural-equivalence}, we formally define structural equivalence and how to incorporate it into a subgraph search. In Section \ref{sec:candidate-equivalence}, we introduce candidate equivalence to demonstrate how to expose previously unseen equivalences during a subgraph search. In Section \ref{sec:node-cover-equivalence}, we introduce node cover equivalence, an alternate form of equivalence which is easy to calculate, and unify all the notions of equivalence into a hierarchy. In Section \ref{sec:experiments}, we adapt the Glasgow solver \cite{mccreesh2020glasgow} to incorporate equivalences and apply it to a set of benchmarks to assess the performance of each of the equivalence levels\footnote{Our implementation of our algorithms can be found at the following repository: https://github.com/domyang/glasgow-subgraph-solver.}.
In Section \ref{sec:compact-solution-representation}, we demonstrate how to succinctly represent and visualize large classes of solution by incorporating equivalence.
In Section \ref{sec:multiplex}, we extend our algorithm to be able to handle multiplex \changedOne{multigraphs} and show our algorithm's success in fully mapping out the solution space on a variety of these more structured networks. 

\changedOne{This paper takes inspiration for the general subgraph tree search structure from \cite{moorman2021subgraph} and shares many of the same test cases on multichannel networks.  
In our previous work \cite{nguyen2019applications}, we introduced a simpler notion of candidate equivalence and candidate structure, and tested it on a small selection of multichannel networks. From these prior works, we observed the high combinatorial complexity of the solution spaces necessitating an approach which can exploit symmetry to compress the solution space and accelerate search. In this work, we expand on both papers by introducing several new notions of equivalence, providing a rigorous foundation for their efficacy in subgraph search, establishing a compact representation of the solution space, and empirically assessing these methods on a broad collection of both real and synthetic data sets.}

\section{Structural Equivalence}\label{sec:structural-equivalence}
\textbf{Structural equivalence} is an easily understood property of networks which, if present, can be exploited to greatly speed up subgraph search. Intuitively, two vertices are structurally equivalent to each other if they can be ``swapped'' without changing the graph structure. This type of equivalence often occurs in leaves that are both adjacent to the same vertex.

\begin{definition}\label{def:structural-equivalence}
In a graph $G = (V, E)$, we say that two vertices $v, w$ are \textbf{structurally equivalent} (denoted $v \sim_s w$) if:
\begin{enumerate}
    \item For $u \in V, u \ne v, w$,
    \begin{enumerate}
        \item $(u, v) \in E \Leftrightarrow (u, w) \in E$
        \item $(v, u) \in E \Leftrightarrow (w, u) \in E$
    \end{enumerate}
    \item $(v, w) \in E \Leftrightarrow (w, v) \in E$
\end{enumerate}
\end{definition}
This definition implies that the neighbors of structurally equivalent vertices (not including the vertices themselves) must coincide. The following proposition verifies that this is an equivalence relation.
%If the graphs are labeled, then in addition we require that $L_V(v) = L_V(w)$ and that the edge labels of each pair of edges above agree.

\begin{prop}\label{prop:structural-equivalence}
$\sim_s$ is an equivalence relation.
\end{prop}
\begin{proof}
    All proofs for propositions stated in this paper are provided in the appendix.
\end{proof}
% \begin{proof}
%     Reflexivity and symmetry are both obvious, so we just need to check transitivity. This is immediate by seeing for $u \sim_s v$ and $v \sim_s w$, it is clear that any $z$ that is an in-neighbor or out-neighbor for $u$, is one for $v$ and therefore is one for $w$. Similarly, if $(u, w) \in E$, then $(v, w) \in E$ implying $(w, v) \in E$ and therefore $(w, u) \in E$. Hence this relation is transitive.
% \end{proof}

Using this relation, we can partition the vertices of any graph into structural equivalence classes, and interchange members of each class without changing the essential structure of the graph. Checking for equivalence between two vertices simply amounts to comparing neighbors in an $O(|V|)$ operation in the worst case, but is generally faster for sparse graphs. Computing the classes themselves can be found by pairwise comparison of vertices resulting in $O(|V|^3)$ operations in the worst case. Algorithm \ref{alg:eq_class_alg} demonstrates how one could implement a breadth first search algorithm to take advantage of the sparsity of a graph to accelerate the computation. Since  For each vertex $v$ visited, it takes $O(deg(v)^2)$ to partition the neighbors of $v$ into equivalence classes, and so in the worst case the algorithm takes $O(\sum_v deg(v)^2) \approx O(|V|\overline{deg(v)^2})$ where $\overline{deg(v)^2}$ denotes the average over $deg(v)^2$ for all $v$. For sparse graphs, $deg(v) \ll |V|$, and so this will be significantly faster than a naive pairwise check.

\begin{algorithm}
\caption{Routine for computing equivalence classes}
\label{alg:eq_class_alg}
\begin{algorithmic}[1]
\Function{FindEqClasses}{$G = (V, E)$}
    \State Let $Q$ be a queue
    \State Pick first vertex $v$ to put in $Q$
    \State Let $EQ = \{\}$ 
    \State Let $\text{visited} = \{\}$
    
    \While{$Q$ not empty}
        \State Dequeue $v$ from $Q$
        \State Add $v$ to visited
        \State Partition $N(v)$ into equivalence classes, to $EQ$
        \State Add representatives from neighbor classes to $Q$
    \EndWhile

    \State Check if first vertex $v$ is in any class and add if so
    \State Else add it to its own class
    \State \Return $EQ$
    \EndFunction
\end{algorithmic}
\end{algorithm}

\subsection{Interchangeability and Isomorphism Counting}

We now show that given any subgraph isomorphism, we can interchange any two vertices in the template graph and still retain a subgraph isomorphism. Before we do this, we formally define what we mean by interchangeability.

\begin{definition}\label{def:interchangeability}
    Two template graph vertices $v, w \in V_T$ are \textbf{interchangeable} if for \changedOne{all} subgraph isomorphism\changedOne{s} $f: V_T \rightarrow V_W$, the \changedOne{mapping} $g$ given by interchanging $v$ and $w$:
    \[
        g(u) = \begin{cases}
        f(w) & u = v \\
        f(v) & u = w \\
        f(u) & \text{otherwise}
    \end{cases}
    \]
    is also a subgraph isomorphism.
    
    Two world graph vertices $v', w' \in V_W$ are \textbf{interchangeable} if for \changedOne{all} subgraph isomorphism\changedOne{s} $f$, if both $v', w'$ are in the image of $f$ with preimages $v, w$, the \changedOne{mapping} $g$:
    \[
        g(u) = \begin{cases}
            w' & u = v \\
            v' & u = w \\
            f(u) & \text{otherwise}
        \end{cases}
    \]
    is an isomorphism. If only one, say $v'$, is in the image, then $h$ given by
    \[
        h(u) = \begin{cases}
            w' & u = v \\
            f(u) & \text{otherwise} 
        \end{cases}
    \]
    is also an isomorphism.

\end{definition}
We will qualify this definition later in the paper by restricting the interchangeability only to certain subsets of isomorphisms. 
The proposition affirming template vertex interchangeability under template structural equivalence follows:

\begin{prop}\label{prop:template-swappability}
    Given graphs $G_T = (V_T, E_T)$, $G_W = (V_W, E_W)$, if $v, w \in V_T$ are structurally equivalent, then they are interchangeable in any subgraph isomorphism.
\end{prop}
% \begin{proof}
%     Obviously $g$ is still injective, so we need only check that it is edge-preserving. Suppose $(x, y) \in E_T$. We consider multiple cases. If neither $x, y$ are $v$ or $w$, then $(g(x), g(y)) = (f(x), f(y)) \in E_W$ as $f$ is a subgraph isomorphism. If $x = v$ and $y=w$, then $(g(x), g(y)) = (f(w), f(v))$. Now as $v \sim_s w$ and $(v, w) \in E_T$, so is $(w, v)$ and using that $f$ is a subgraph isomorphism, $(f(w), f(v)) \in E_W$. Now suppose one of $x,y$ is one of $v,w$. Without loss of generality let $x = v$. Then $(g(x), g(y)) = (f(w), f(y))$. As $v \sim_s w$, $y$ is also an out-neighbor of $w$, and so $(g(x), g(y)) = (f(w), f(y)) \in E_W$. This verifies that $g$ is a subgraph isomorphism.
% \end{proof}

Hence, interchanging the images of two template vertices preserves subgraph isomorphism.
As transpositions generate the full set of permutations, we have the following result:
\begin{prop}\label{prop:template-eq-counting}
If we can partition $V_T = C_1,\ldots,C_n$ into structural equivalence classes, and there exists at least one subgraph isomorphism, then there are at least
\[
    \prod_{i=1}^n |C_i|!
\]
subgraph isomorphisms.
\end{prop}
% \begin{proof}
%     Using one subgraph isomorphism, we can permute the elements of each class $C_i$, of which there are $|C_i|!$ permutations. Multiplying these numbers together gives the full count.
% \end{proof}

We can also apply this structural equivalence to the world graph to demonstrate a similar kind of interchangeability.

\begin{prop}\label{prop:world-swappability-1}
    If $v', w' \in V_W$ are structurally equivalent, then in any subgraph isomorphism $f$, they are interchangeable. 
\end{prop}
% \begin{proof}
% The proofs of this proposition for the two types of interchanges in Definition \ref{def:interchangeability} are very similar, so we just prove the first one. $g$ is obviously injective, so we just check that it is edge-preserving. Given $(x,y) \in E_T$, we consider three cases. In the first case, neither of $x$ and $y$ are $v$ or $w$, $\{x, y\} \cap \{v, w\} = \emptyset$. Then $(g(x), g(y)) = (f(x), f(y)) \in E_W$. In the second case, one of $x$ and $y$ equals one of $v$ or $w$; either $x \in \{v, w\}$ or $y \in \{v, w\}$. Without loss of generality, we take $x=v$, then $(g(x), g(y)) = (w', f(y))$. As $w' \sim_s v'$ and $(f(x), f(y)) = (v', f(y)) \in E_W$, we have $(w', f(y)) \in E_W$. For the final case,  $\{x,y\} = \{v, w\}$. Without loss of generality we take $x=v, y=w$. Then $(g(x), g(y)) = (w', v') = (f(y), f(x))$. As $(f(v), f(w)) = (v', w') \in E_W$, the reverse edge $(w', v') \in E_W$ since $v' \sim_s w'$. This completes the proof.
% \end{proof}

If we apply both template and world structural equivalence to our problem, it is natural to ask how many solutions we can now generate from a single solution. This is given by the following proposition:
\begin{prop}\label{prop:tewe-counting}
    Let $f:V_T \rightarrow V_W$ be a subgraph isomorphism.
    Suppose we partition the template graph $V_T = \bigcup_{i=1}^nC_i$ and the world graph $V_W = \bigcup_{j=1}^mD_j$ into structural equivalence classes. Let $C_{i,j} = C_i \cap f^{-1}(D_j)$ represent the set of template vertices in $C_i$ that map to world vertices in the equivalence class $D_j$. Then there are 
    \[
        \prod_{i=1}^n|C_i|!\prod_{j=1}^m\prod_{i=1}^n\binom{|D_j| - \sum_{k=1}^{i-1}|C_{k,j}|}{|C_{i,j}|}
    \]
    isomorphisms generated by interchanging equivalent template vertices or world vertices using Propositions \ref{prop:template-swappability} and \ref{prop:world-swappability-1}.
\end{prop}
% \begin{proof}
%     As before, the first factor comes from all permutations on the equivalence classes. 
%     Once we fix a permutation, from each world equivalence class $D_j$, we need to choose $|C_{i,j}|$ elements to be the values for the elements of $C_{i,j}$. For $C_{1,j}$, we have 
%     $\binom{|D_j|}{|C_{1,j}|}$ elements, for $C_{2,j}$, we have 
%     $\binom{|D_j|-|C_{1,j}|}{|C_{2,j}|}$ as we have already used $|C_{1,j}|$ elements, and so on. Taking the product over $j$ gives the second factor.
% \end{proof}

\subsection{Application to Tree Search}\label{subsec: application to tree search}
We now demonstrate how to adapt any tree-search algorithm to incorporate equivalence. A tree-search algorithm proceeds by constructing a partial matching of template vertices to world vertices, at each step extending the matching by assigning the next template vertex to one of its candidate world vertices. If at any point, the match cannot be extended (due to a contradiction or finding a complete matching), the last assigned template vertex is reassigned to the next candidate vertex. Each possible assignment of template vertex to world vertex corresponds to a node in the tree, and a path from the root of the tree to a leaf corresponds to a full mapping of vertices.

The tree search is a recursive routine described by Algorithm \ref{alg:tree-search}. In this procedure, we maintain a binary $\abs{V(T)} \times \abs{V(W)}$ matrix $cands$ where $cands[i,j]$ is 1 if world vertex $j$ is a candidate for template vertex $i$ and 0 otherwise and a mapping from template vertices to world vertices describing which vertices have already been matched. In lines 2-4, we report a match after having matched all vertices. The call to ApplyFilters in line 5 eliminates candidates based on the assumptions made so far in the partial match. In line 7, we save the current state to return to after backtracking, and in lines 11-14, we iterate through all candidates for the current vertex attempting a match until we have exhausted them all. Then in line 18, we restore the prior state and backtrack.

\begin{algorithm}
\caption{Generic routine for a tree search}\label{alg:tree-search}
    \begin{algorithmic}[1]
    \Function{Solve}{partial\_match, cands}
    \If{MatchComplete(partial\_match)}
        \State{ReportMatch(partial\_match)}
        \State \Return
    \EndIf
    
    \State ApplyFilters(partial\_match, cands)
    \State Let $u$ = GetNextTemplateVertex()
    \State Let cands\_copy = cands.copy()
    \If{Using World Equivalence}
        \State RecomputeEquivalence(partial\_match, cands)
    \EndIf 
    
    \State Let ws = GenerateWorldVertices(cands, eq)
    
    \For{$v$ in ws}
        \State partial\_match.match$(u,v)$
        \State Solve(partial\_match, cands\_copy)
        \State partial\_match.unmatch$(u,v)$
        \If{Using Template Equivalence}
            \For{unmatched $u'\sim u$}
                \State Set cands$[u',v]$ = 0
            \EndFor
        \EndIf
    \EndFor
    \State Let cands = cands\_copy
    \If{Using World Equivalence}
        \State RestoreEquivalence()
    \EndIf  
    
    \State \Return
    \EndFunction
    \end{algorithmic}
\end{algorithm}

Template equivalence can significantly accelerate the tree search; from Proposition \ref{prop:template-swappability} we can swap the assignments of equivalent template vertices to find another isomorphism. If we have a partial match, template vertices $u_1 \sim u_2$, and we have just considered candidate $w$ for $u_1$, we can ignore branches where $u_2$ is mapped to $w$ since we can generate those isomorphisms by taking one where $u_1$ is mapped to $w$ and swapping. Lines 16 and 17 demonstrate how we can incorporate this idea into a tree search (without these, we would have a standard tree search).

To incorporate world equivalence into the search, we modify the search so that we only assign any template vertex to one representative of an equivalence class in the search. This can be done by modifying GenerateWorldVertices to pick only one representative vertex of each equivalence class out of the candidates for the current template vertex. 

Note that after performing the tree search, the solutions found will represent classes of solutions that can be generated by swapping. Some bookkeeping is needed to determine what assignments can be swapped to count the number of distinct solutions.
We call the solutions that are actually found (and are not produced by interchanging vertices) \textbf{representative solutions}. The set of solutions that can be generated by interchanging equivalent vertices for a given representative solution is a \textbf{solution class}. The ability to represent large solution classes with a sparse set of solutions is what allows us to compactly describe massive solution spaces.

\section{Candidate Equivalence}\label{sec:candidate-equivalence}
The equivalence discussed in the prior section is a static form of equivalence, only taking into account information provided at the start of the subgraph search. However, as a subgraph search proceeds, we may be able to discard additional nodes and edges based on information derived from the assignments already made. For example, in \figref{fig:toy-example}, after assigning A to 1, we may discard nodes 6 and 7 and edge (4, 5), as it is impossible for them be included in a match if A and 1 are matched. After these nodes are discarded, we discover that with respect to the matches already made, 2, 3, 4, and 5 are effectively interchangeable.
In order to make use of this dynamic form of equivalence, we need to introduce an auxiliary structure that takes into account our knowledge of the candidates of each vertex $u$ (denoted $C[u]$).
\begin{definition}\label{def:candidate-structure}
    Given template graph $G_T = (V_T, E_T)$, world graph $G_W = (V_W, E_W)$, and candidate sets $C[u] \subset V_W$ for each $u \in V_T$, the \textbf{candidate structure} is the directed graph $G_C = (V_C, E_C)$ where the vertices $V_C = \{(u, c): u \in V_T, c \in V_W\}$ are template vertex-candidate pairs and $((u_1, c_1), (u_2, c_2)) \in E_C$ if and only if $(u_1, u_2) \in E_T$ and $(c_1, c_2) \in E_W$.
\end{definition}

The candidate structure represents both the knowledge of candidates for each template node and how template adjacency interplays with world adjacency. It removes extraneous information to expose equivalences not apparent when looking at the original graphs. \changedOne{We note that this data structure is similar to the compact path index (CPI) introduced in \cite{bi2016efficient}. However, the CPI is only defined for a given rooted spanning tree of the template graph whereas our candidate structure takes into consideration all edges of the template graph.}
For our toy example in \figref{fig:toy-example}, if we assume that the candidate sets are reduced to the minimal candidate sets so that $C[A] = \{1, 4\}$ and $C[B] = \{2,3,4,5,6,7\}$. Then the candidate structure for these graphs is as shown on the left in \figref{fig:example-candidate-structure}.
\begin{figure}
    \centering
    \includegraphics[height=100px]{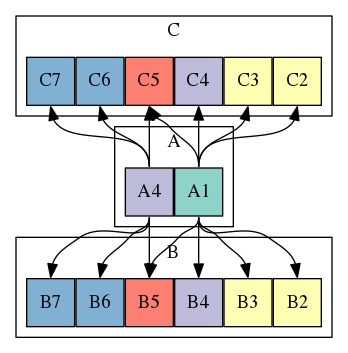}
    \includegraphics[height=100px]{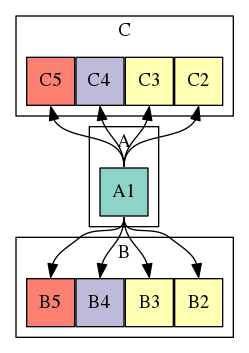}
    \caption{Candidate structure for the graphs in Figure \ref{fig:toy-example} before and after assigning template node A to world node 1.}
    \label{fig:example-candidate-structure}
\end{figure}
At this point, there is no apparent equivalence to exploit from the candidate structure. However once we decide to map node A to 1, the candidate structure reduces to the right graph in \figref{fig:example-candidate-structure}. It is visually clear that nodes 2, 3, 4, and 5 are structurally equivalent as candidates of B and C. Similarly, if we assigned A to 4, nodes 5, 6, and 7 will be structurally equivalent in the candidate structure. We want to determine under which circumstances this will ensure interchangeability.  We introduce the following definition:
\begin{definition}\label{def:candidate-equivalence}
    Given a candidate structure $G_C = (V_C, E_C)$, $c_1, c_2 \in V_W$, we say that $c_1$ is \textbf{candidate equivalent} to $c_2$ with respect to $u \in V_T$, denoted $c_1 \sim_{c,u} c_2$, if and only if $c_1, c_2 \notin C[u]$ or $c_1, c_2 \in C[u]$ and $(c_1, u) \sim_s (c_2, u)$.
\end{definition}
It is easy to show that if the candidate sets are \textit{complete} (for each template node $u$, if there is a matching which maps $u$ to world node $v$, then $v \in C[u]$), then if $c_1 \sim_s c_2$, then $c_1 \sim_{c,u} c_2$ for all template vertices $u$.

The exact criteria for interchangeability is a little more complicated. For example, in \figref{fig:toy-example}, 4 appears as a candidate for both A and for B and C, so that we cannot simply swap 4 with nodes that are  candidate equivalent to 4 with respect to B. 
\changedOne{To address a more complex notion of interchangeability, we introduce some terms. We say that a subgraph isomorphism $f$ is derived from a candidate structure $G_C$ if for any $v \in V_T, (v, f(v)) \in V_C$ (i.e., $f(v)$ is a candidate of $v$). We say that world vertices $w_1, w_2$ are \textbf{$G_C$-interchangeable} if for all isomorphisms derived from the candidate structure $G_C$, $w_1$ and $w_2$ can be interchanged and preserve isomorphism.}

A simple criterion for interchangeability is provided in the following proposition:
\begin{prop}\label{prop:candidate-swappability}
Suppose that given a specific candidate structure $G_C = (V_C, E_C)$, we have that $c_1,c_2 \in C[u]$ and $c_1 \sim_{c,u} c_2$ for some template vertex $u$. Suppose that $c_1$ and $c_2$ are not candidates for any other vertex. Then $c_1$ and $c_2$ are \changedOne{$G_C$-interchangeable}.
\end{prop}
% \begin{proof}
%     This is clearly injective, and for any edge $(v, w)$ which doesn't include $u$, $g$ agrees with $f$, so it preserves those edges. If we have $(u, v) \in E_T$, we have $(f(u), f(v)) = (c_1, f(v)) \in E_W$. Since $c_1 \sim_{c,u} c_2$, and we have $((u, c_1), (v, f(v)) \in E_C$, we must have $((u, c_2), (v, f(v)) \in E_C$ which implies $(c_2, f(v)) \in E_W$.
% \end{proof}

This proposition suggests a simple method for exploiting candidate equivalence. In our tree search, when we generate candidate vertices for a given vertex $u$, we find representatives, for each candidate equivalence class, that do not appear as candidates for other vertices. If a class has a vertex appearing in other candidate sets, then we cannot exploit equivalence and must check each member of the class. Furthermore, as we continue to make matches and eliminate candidates, more world vertices will become equivalent, so it is advantageous to recompute equivalence before every match as is done in line 9 of Algorithm \ref{alg:tree-search}. Upon unmatching, we need to restore the prior equivalence in line 20.

If we have that $f(v) = c_1$ and $f(w) = c_2$, and we want to swap $c_1$ for $c_2$, we need a stronger condition; namely, we need that they are equivalent with respect to both $v$ and $w$. In the process of a tree search, we do not know exactly what each vertex will be mapped to so instead we consider an even stronger condition:
\begin{definition}\label{prop:full-equivalence}
    Given a candidate structure $G_C = (V_C, E_C)$, we say that $c_1 \in V_W$ is \textbf{fully candidate equivalent} to $c_2 \in V_W$, denoted $c_1 \sim_c c_2$ if for all $u \in V_T$, $c_1 \sim_{c,u} c_2$.
\end{definition}
Note that if $c_1 \sim_{c,u} c_2$ for some $u$, and $c_1,c_2$ are not candidates for any other vertices, then $c_1 \sim_c c_2$. This condition enables us to interchange world vertices and still maintain the subgraph isomorphism conditions. This is established by the following proposition:

\begin{prop}\label{prop:full-swappability}
    Suppose that given a specific candidate structure $G_C = (V_C, E_C)$, we have that $c_1,c_2 \in V_W$ and $c_1 \sim_{c} c_2$. Then, $c_1$ and $c_2$ are \changedOne{$G_C$-interchangeable}.
\end{prop}
% \begin{proof}
% If $(x, y) \in E_T$, and neither is $u_1$ or $u_2$, then $g$ agrees with $f$ and preserves the edge. If one is $u_1$ or $u_2$, without loss of generality take $x = u_1$, then $(g(x), g(y)) = (c_2, f(y))$. As $f$ is a subgraph isomorphism $(c_1, f(y)) \in E_W$ and since $c_1 \sim_c c_2$, $(c_2, f(y)) \in E_W$ as well. If $x = u_1, y=u_2$, then $(g(x), g(y)) = (c_2, c_1) = (f(y), f(x))$ and this edge is in $E_W$ since $c_1 \sim_c c_2$ and $(c_1, c_2) \in E_W$.
% \end{proof}

\section{Node Cover Equivalence}\label{sec:node-cover-equivalence}
An alternate notion of equivalence, introduced in \cite{moorman2021subgraph}, involves the use of a node cover. 
A \textbf{node cover} is a subset of nodes whose removal, along with incident edges, results in a completely disconnected graph. 
The approach in \cite{moorman2021subgraph} is to build up a partial match of all the vertices in the node cover followed by assigning all the nodes outside the node cover. 
After reducing the candidate sets of all the nodes outside the cover to those that have enough connections to the nodes in the cover, what remains is to ensure that they are all different.

We formalize this with some definitions. A \textbf{partial match} is a subgraph isomorphism from a subgraph of the template graph to the world graph. We list out the mapping as \changedOne{as a list of ordered pairs} $M = \{(v_1, w_1),\ldots,(v_n, w_n)\}$. A template vertex - candidate pair $(v, c)$ is \textbf{joinable} to a partial match $M$ if for each $(v_i, w_i) \in M$, if $(v_i, v) \in V_T$, then $(w_i, c) \in V_W$ and if $(v, v_i) \in V_T$, then $(c, w_i) \in E_T$. \changedOne{If two world vertices $w_1, w_2$ are interchangeable in any subgraph isomorphism extending a partial match $M$, we say that $w_1$ and $w_2$ are \textbf{$M$-interchangeable}}.

Since the problem is significantly simpler, it is easier to obtain a form of equivalence on the vertices.
\begin{definition}\label{def:node-cover-equivalence}
Let $M$ be a partial match $M$ on a node cover $N$ of $V_T$ and suppose that for all $u \in V_T \setminus N$, the candidate set $C[u]$ is comprised entirely of all world vertices joinable to $M$. Two world vertices $w_1, w_2$ are \textbf{node cover equivalent} with respect to $M$, denoted $w_1 \sim_{N,M} w_2$, if for all $u \in V_T \setminus N$, $w_1 \in C[u]$ if and only if $w_2 \in C[u]$.
\end{definition}

For example, consider the template and world in \figref{fig:node-cover-eq}. Once nodes B and D in the node cover are mapped to 2 and 5, the remaining nodes have candidates that have the associated color in the world graph. We then simply group each of these candidates together into equivalence classes. Note that the edges depicted in red are what prevent structural equivalence, and the node cover approach effectively ignores these edges to expose the equivalence of these vertices.

\begin{figure}
    \centering
    \includegraphics[height=100px]{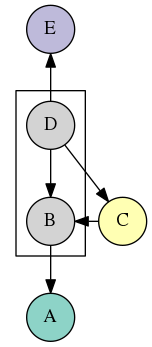}
    \includegraphics[height=100px]{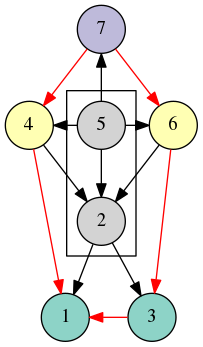}
    \includegraphics[height=100px]{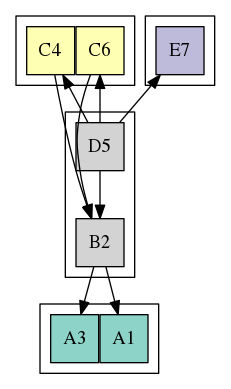}
    \caption{In order from left to right: template, world, and possible candidate structure. The boxed vertices comprise a node cover of the template and the image of the node cover in the world. Nodes of the same color in the world are node cover equivalent. The red edges are extraneous edges which once removed, expose equivalence.}
    \label{fig:node-cover-eq}
\end{figure}

\begin{prop}\label{prop:node-cover-swappability}
Suppose we have a node cover of the template graph $N$, and a partial matching $M$ on $N$ and two world vertices $w_1, w_2$ not already matched satisfy $w_1 \sim_{N,M} w_2$. \changedOne{Then $w_1$ and $w_2$ are $M$-interchangeable}.
\end{prop}
% \begin{proof}
%     Let $f$ be such an isomorphism which maps $u_1$ to $w_1$ and $u_2$ to $w_2$ and $g$ interchanges $w_1$ and $w_2$.
%     Consider $(x,y) \in E_T$. If neither are $u_1, u_2$, then $g$ agrees with $f$ and so the edge is preserved. If one of them is $u_1$ or $u_2$, say $x = u_1$, then it must be that $y$ is in $N$ as $N$ is a node cover ($u_1$ is disconnected from any element outside the node cover). Since $f$ is a subgraph isomorphism, $(u_1, w_1)$ must be joinable to $M$ and $(w_1, f(y)) \in E_T$ and so $w_1 \in C[u_1]$. It must be that $w_2 \in C[u_1]$ and so $(u_1, w_2)$ is also joinable to $M$. Hence $(w_2, f(y)) = (g(x), g(y)) \in E_T$.
%     The last case $x = u_1$ and $y = u_2$  is impossible since $x$ and $y$ are outside the node cover and therefore disconnected.
% \end{proof}

Node cover equivalence is easy to check and captures a significant portion of the equivalence posed by other methods. 
This is often due to interchangeable nodes being composed of sibling leaves which are generally outside of a node cover.

We note that the methods discussed in the paper (with the exception of basic structural equivalence for template and world as in Proposition \ref{prop:tewe-counting}) cannot incorporate both template and world equivalence.
The combination of allowing template and world node interchanges and having dynamic world equivalence classes significantly complicates the counting process.
One approach which can facilitate the use of both forms of equivalence involves a tree search where entire template equivalence classes are assigned at once instead of individual template nodes.
This kind of approach for assigning the remaining nodes outside of the node cover is discussed in the Appendix Section D.

\subsection{Equivalence Hierarchy}

There is a relation between node cover equivalence and full candidate equivalence, given in the following proposition:
\begin{prop}\label{prop:node-cover-full-cand-equiv}
    Suppose that $N$ is a node cover of $V_T$, $M$ is a partial match on $N$, candidate sets are reduced to joinable vertices, and $w_1, w_2 \in V_T\setminus N$. Then $w_1 \sim_{N,M} w_2 \Leftrightarrow w_1 \sim_c w_2$. 
\end{prop}
% \begin{proof}
%     Fix a template vertex $u$. If $u \in N$, then these nodes are already assigned to world vertices, neither of which will be $w_1$ or $w_2$ and so $w_1, w_2 \notin C[u]$. Therefore $w_1 \sim_{c,u} w_2$. If $u \notin N$, and we have $((u, w_1), (v, x)) \in E_C$, then $(u, v) \in E_T$ and $(w_1, x) \in E_W$. $v$ must be inside the node cover as those can be the only connections to $u$ and therefore $v$ must already be assigned to $x$. As $w_1 \sim_{N,M} w_2$, we must have $w_2 \in C[u]$ as well, and so must be joinable to the matching. This implies that $(w_2, x) \in E_W$. Hence we must have $((u, w_2), (v, x)) \in E_C$. Since this edge was chosen arbitrarily, we must have $w_1 \sim_{c,u} w_2$. Since this holds for all $u$, we must have $w_1 \sim_c w_2$.
    
%     On the other hand, if $w_1 \sim_c w_2$, for any template vertex $t$, if $(t, w_1)$ is joinable to $M$, then $(t, w_2)$ is also joinable to $M$. Hence, the template nodes for which $w_1$ and $w_2$ are candidates coincide, so that $w_1 \sim_{N,M} w_2$.
% \end{proof}

Thus, until we have assigned a node cover, we can use candidate equivalence to prevent redundant branching; once we have matched all nodes in the node cover, we can check for node cover equivalence: a simpler condition.

The agreement of node cover equivalence and fully candidate equivalence is apparent in the candidate structure presented on the right of \figref{fig:node-cover-eq}. From the candidate structure, the yellow nodes and the green nodes are fully candidate equivalent as they have the same neighbors, and that they are node cover equivalent as they only appear as candidates for the corresponding yellow and green nodes in the template.

The various notions of equivalence form a hierarchy. Structural equivalence of world nodes has the strictest requirements and implies all other forms of equivalence. 
Proposition \ref{prop:node-cover-full-cand-equiv} asserts that under mild conditions, full candidate equivalence and node cover equivalence are one and the same. We can also include candidate equivalence with respect to a template vertex as a weaker condition implied by full candidate equivalence that does not guarantee interchangeability.
The following proposition summarizes these findings:
\begin{prop}\label{prop:equiv-hierarchy}
    Suppose the assumptions of Proposition \ref{prop:node-cover-full-cand-equiv} hold. Given template vertex $t$, world vertices $w_1, w_2$, we have $w_1 \sim_s w_2 \Rightarrow w_1 \sim_{N,M} w_2  \Leftrightarrow w_1 \sim_{c} w_2 \Rightarrow w_1 \sim_{c,t} w_2$. Under the first three equivalences, $w_1$ and $w_2$ are interchangeable.
\end{prop}

From this proposition, we observe that full candidate equivalence and node cover equivalence provide the most compact solution space as they require the weakest conditions while still guaranteeing interchangeability.
However, the cost in determining full candidate equivalence may be prove excessive compared to simpler types of equivalence. We address these trade-offs on real and simulated data in Section \ref{sec:experiments}.

\section{Experiments}\label{sec:experiments}

To demonstrate the utility of equivalence for the SMP, we adapt a state-of-the-art tree search subgraph isomorphism solver, Glasgow \cite{mccreesh2020glasgow}, using the modifications described in Algorithm \ref{alg:tree-search}. We consider seven  levels of equivalence: no equivalence (NE) (default), template structural equivalence (TE), world structural equivalence (WE), template and world structural equivalence (TEWE), candidate equivalence as in Proposition \ref{prop:candidate-swappability} (CE), full candidate equivalence as in Proposition \ref{prop:full-swappability} (FE), and node cover equivalence (NC). Each equivalence mode is integrated into the Glasgow solver separately.

For each test class involving template equivalence (TE, TEWE), we compute the template structural equivalence classes, and for each involving world equivalence (WE, TEWE, CE, FE, NC), we compute world structural equivalence classes at the start using Algorithm \ref{alg:eq_class_alg}. Then we make the modifications for template and world equivalence as in Algorithm \ref{alg:tree-search}. For algorithms requiring recomputation of the equivalence classes at each node of the tree search, for speed purposes, we only recompute equivalence for nodes that appear as candidates for the current template node under consideration. We check equivalence between each pair of candidates using the definitions directly.

\begin{table*}
\centering
\caption{Benchmark Dataset Statistics}
\label{tab:benchmark-statistics}
\begin{tabular}{@{}l|l|llllll|llllll@{}}
\toprule
    &    & \multicolumn{6}{c}{Template}                                                              & \multicolumn{6}{c}{World}                                                                 \\ 
    Dataset        & \# Instances & \multicolumn{2}{c}{\# Nodes} & \multicolumn{2}{c}{\# Edges} & \multicolumn{2}{c}{Density} & \multicolumn{2}{c}{\# Nodes} & \multicolumn{2}{c}{\# Edges} & \multicolumn{2}{c}{Density}  \\ 
 &  & Min & Max & Min & Max & Min & Max & Min & Max & Min & Max & Min & Max \\
            \midrule
SF          & 100          & 180          & 900           & 478         & 5978           & 0.006        & 0.165        & 200           & 1000         & 592          & 7148          & 0.006        & 0.159        \\
LV          & 6105         & 10           & 6671          & 10          & 209000         & 0.001        & 1.000        & 10            & 6671         & 10           & 209000        & 0.001        & 1.000        \\
SI          & 1170         & 40           & 777           & 41          & 12410          & 0.005        & 0.209        & 200           & 1296         & 299          & 34210         & 0.004        & 0.191        \\
images-cv   & 6278         & 15           & 151           & 20          & 215            & 0.019        & 0.190        & 1072          & 5972         & 1540         & 8891          & 4.89e-4        & 0.003        \\
meshes-cv   & 3018         & 40           & 199           & 114         & 539            & 0.022        & 0.146        & 201           & 5873         & 252          & 15292         & 4.40e-4        & 0.022        \\
images-pr   & 24           & 4            & 170           & 4           & 241            & 0.017        & 0.667        & 4838          & 4838         & 7067         & 7067          & 0.001        & 0.001        \\
biochemical & 9180         & 9            & 386           & 8           & 886            & 0.012        & 0.423        & 9             & 386          & 8            & 886           & 0.012        & 0.423        \\
phase       & 200          & 30           & 30            & 128         & 387            & 0.294        & 0.890        & 150           & 150          & 4132         & 8740          & 0.370        & 0.782        \\ 
www & 3850 & 5 & 15 & 5 & 45 & 0.071 & 0.750 & 325729 & 325729 & 1497135 & 1497135 & 1.41e-5& 1.41e-5 \\ \bottomrule
\end{tabular}

\end{table*}

We consider single-channel networks from the benchmark suite in \cite{solnon2019experimental}.  Basic parameters for the datasets are listed in Table \ref{tab:benchmark-statistics}. \textit{SF} is composed of 100 instances that are randomly generated using a power law and are designed to be scale-free networks. \textit{LV} is a diverse collection of randomly generated graphs satisfying various properties (connected, biconnected, triconnected, bipartite, planar, etc.). \textit{SI} is a collection of randomly generated instances falling into four categories: bounded valence, modified bounded valence, 4D meshes, and Erdős–Rényi graphs. The \textit{images-cv}, \textit{meshes-cv}, and \textit{images-pr} data \cite{damiand2011polynomial, solnon2015complexity} sets are real instances representing segmented images and meshes of 3D objects drawn from the pattern recognition literature. The \textit{biochemical} dataset \cite{gay2014subgraph} \changedOne{contains} matching problems \changedOne{taken from real} systems of biochemical reactions. The \textit{phase} dataset \cite{mccreesh2018subgraph} \changedOne{is comprised of} randomly generated Erdős–Rényi graphs with parameters  known to be very difficult for state-of-the-art solvers. 

\changedOne{We also include a problem set where the template graph is a small Erdős–Rényi graph and the world graph is composed of the webpages on the Notre Dame university website with directed edges representing links between pages \cite{albert1999diameter}.
In these instances, the template is randomly generated with $n_t$ nodes and $e_t$ edges where $5 \le n_t \le 15$ and $n_t \le e_t \le 3n_t$. The world graph is fairly sparse and has 325,729 vertices and 1,497,135 edges. We refer to this problem set as the \textit{www} dataset. We collected 50 template graphs for each value of $n_t$ for a total of 550 templates.}

For each instance, we run the algorithm for each equivalence level with the solver configured to count all solutions.  We record the number of representative solutions found, the total number of solutions that can be generated by interchanging, as well as the total run time for the instance. We terminate each run if the search is not completed after 600 seconds and record the statistics for the incomplete run. For each run, we measure the compression rate which is the number of representative solutions divided by the total number of solutions found. This quantity indicates the factor by which the form of equivalence chosen decreases the size of the solution space. \changedOne{All experiments were performed on an Intel Xeon Gold 6136 processor with 3 GHz, 25 MB of cache, and 125 GB of memory.}

\begin{table}[]
\caption{Proportion of Satisfiable Instances Fully Enumerated}
\centering
\label{tab:proportion-enum}
\begin{tabular}{@{}l|lllllll@{}}
\toprule
Dataset     & NE   & TE   & WE   & FE   & TEWE & CE   & NC   \\ \midrule
biochemical & 0.73 & 0.78 & 0.83 & 0.86 & 0.81 & 0.84 & 0.85 \\
LV          & 0.10 & 0.10 & 0.14 & 0.15 & 0.13 & 0.14 & 0.15 \\
scalefree   & 1.00 & 1.00 & 1.00 & 1.00 & 1.00 & 1.00 & 1.00 \\
images-cv   & 1.00 & 1.00 & 1.00 & 1.00 & 1.00 & 1.00 & 1.00 \\
meshes-cv   & 0.00 & 0.00 & 0.00 & 0.00 & 0.00 & 0.00 & 0.00 \\
si          & 0.83 & 0.92 & 0.83 & 0.94 & 0.90 & 0.85 & 0.93 \\
images-pr   & 1.00 & 1.00 & 1.00 & 1.00 & 1.00 & 1.00 & 1.00 \\
phase       & 0.00 & 0.00 & 0.00 & 0.00 & 0.00 & 0.00 & 0.00 \\
www         & 0.00 & 0.00 & 0.00 & 0.00 & 0.00 & 0.00 & 0.00 \\
\bottomrule
\end{tabular}
\end{table}

\begin{figure}
    \centering
    \includegraphics[width=\linewidth]{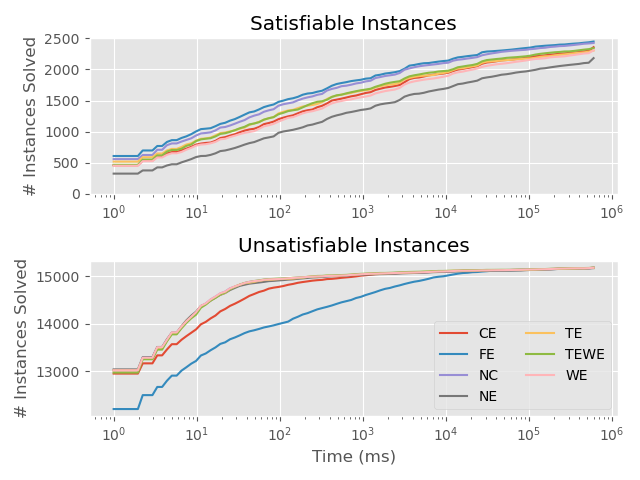}
    \caption{Number of satisfiable (top) and unsatisfiable (bottom) instances solved after a given amount of time. This is aggregated over all single channel benchmark data sets. For satisfiable instances, ``solved'' means having fully enumerated the solution space.}
    \label{fig:solveby}
\end{figure}

In Table \ref{tab:proportion-enum}, we record the proportion of satisfiable problems for which the solution space is fully enumerated by each algorithm within 600 seconds. 
For the \textit{biochemical}, \textit{LV}, and \textit{si} datasets, there is an increase of 5-10\% of all problems fully solved when using any form of equivalence. 
We further note that full equivalence always has the best performance followed by node cover equivalence. This is not surprising given that Proposition \ref{prop:node-cover-full-cand-equiv} states that full equivalence is the most expansive form of equivalence.

\figref{fig:solveby} portrays how many instances are solved by a given time and demonstrates that full equivalence performs best in solution space enumeration for satisfiable problems followed by node cover and the other forms of equivalence. 
The plot for unsatisfiable problems demonstrates drawbacks to using full candidate equivalence; the additional computation time to check equivalence is not needed if there are no solutions. 
Checking equivalence in the TE, WE, and NC routines are very cheap operations so there is little difference in the amount of unsatisfiable problems solved compared to the NE routine.

\begin{figure*}
    \centering
    \includegraphics[width=0.45\linewidth]{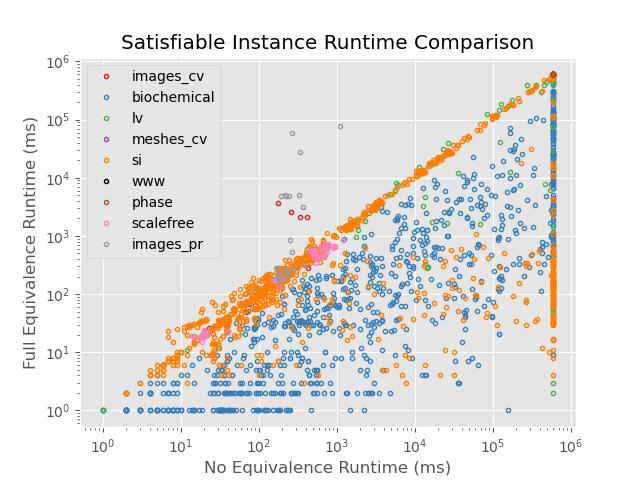}
    \includegraphics[width=0.45\linewidth]{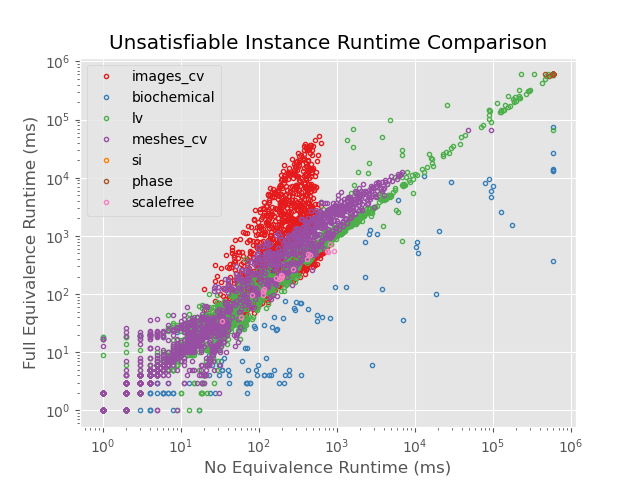}
    \caption{Comparison of individual run times for full enumeration between no equivalence and full equivalence runs for satisfiable problems (left) and unsatisfiable problems (right). \changedOne{Note the \textit{phase}, \textit{www}, and \textit{meshes\_cv} problems do not terminate for any instance and take the full $600$ second runtime, so they can be difficult to discern as each occupies the same spot in the upper right corner of the graphs.}}
    \label{fig:individual-runtimes}
\end{figure*}

\begin{figure*}
    \centering
    \includegraphics[width=0.45\linewidth]{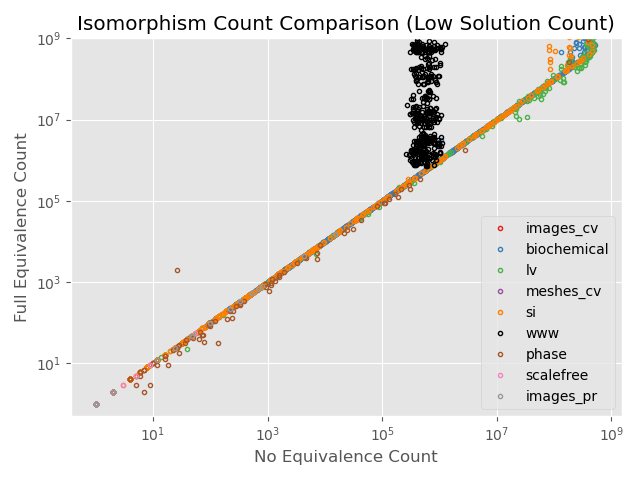}
    \includegraphics[width=0.45\linewidth]{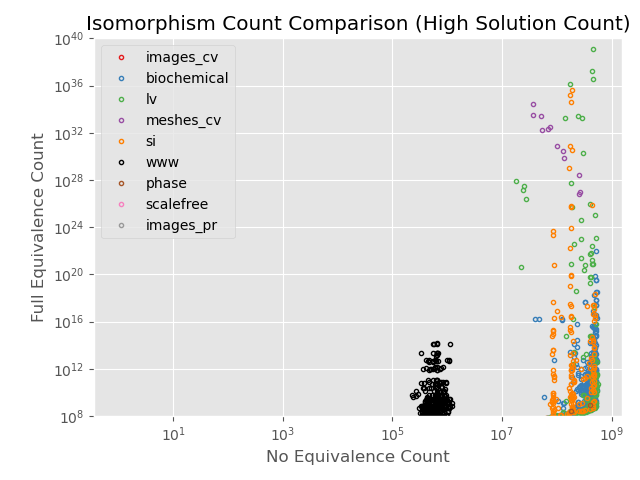}
    \caption{Comparison of isomorphism counts for full enumeration between no equivalence and full equivalence runs for problems with small \changedOne{($<10^9$)} numbers of isomorphisms (left) and problems with large \changedOne{($\ge 10^9$)} numbers of isomorphism (right). Take note of the scales chosen for each graph. For 110 instances the solver with full equivalence found greater than $10^{40}$ isomorphisms (the largest had $\approx 10^{384}$ isomorphisms), and they are not shown on these graphs.}
    \label{fig:individual-iso-counts}
\end{figure*}

\figref{fig:individual-runtimes} demonstrates the variation among and within the data sets by comparing run times for the NE and FE routines. 
We observe that it is primarily among the \textit{biochemical}, \textit{SI}, and \textit{LV} for which the full equivalence routine vastly outperforms the no equivalence routine often by several orders of magnitude.
On the other hand, the \textit{images-cv}, \textit{meshes-cv}, and \textit{images-pr} datasets are more challenging, especially for unsatisfiable problems. This may be due to wasted equivalence checks as there is no solution space to be compressed.

\figref{fig:individual-iso-counts} includes two plots to illustrate variation in solution counts when using the NE and FE routines. 
Each has different limits on the axes to emphasize different aspects. 
The left demonstrates that for problems with fewer than $10^9$ isomorphisms, 10 minutes is often enough time to fully enumerate the solution space without using equivalence. The number $10^9$ functions as an approximate upper bound for the number of solutions found in 10 minutes for any problem using the original Glasgow solver. \changedOne{We note that for the \textit{www} dataset, the base Glasgow solver finds at most approximately $10^6$ solutions, a significantly smaller number. This happens because that a significant portion of time is taken to load in the large world graph.}
The right plot shows problems for which the FE routine finds many orders of magnitude more solutions. In these highly symmetric problems, an equivalence-based approach is essential to fully understand the solution space. \changedOne{We observe that for the \textit{biochemical}, \textit{si}, \textit{lv}, and \textit{www}, we find many orders of magnitude more solutions when using full equivalence.} The largest disparity found between FE and NE solution counts is not displayed - FE found about $10^{384}$ solutions and NE found roughly $10^9$ solutions: a difference of 375 orders of magnitude.

\begin{figure}
    \centering
    \includegraphics[width=\linewidth]{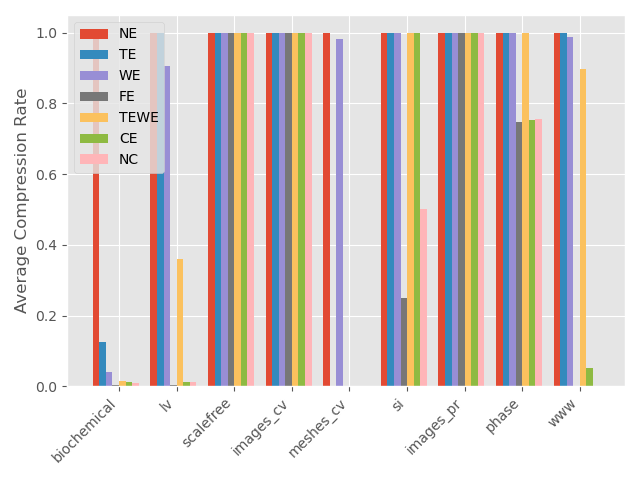}
    \caption{The average compression rate for each dataset and equivalence type.}
    \label{fig:compression-rate}
\end{figure}

\changedOne{Figure \ref{fig:compression-rate} demonstrates the different average compression rates across each dataset and equivalence level. As expected, FE, NC, and CE perform the best in terms of compression followed by TEWE and then depending on the dataset, either TE or WE. For the \textit{biochemical}, \textit{lv}, \textit{meshes\_cv}, and \textit{www} datasets, we observe on average, the solution space is compressed in size by an order of magnitude or more when using the CE, FE, or NC methods. On specific particularly symmetric problems from these datasets, we find the solution space can be compressed by tens or even hundreds of orders of magnitude when using FE or NC. For these cases, it is necessary to incorporate equivalence to come close to understanding the set of solutions to the subgraph isomorphism problem.}

\begin{figure}
    \centering
    \includegraphics[width=0.8\linewidth]{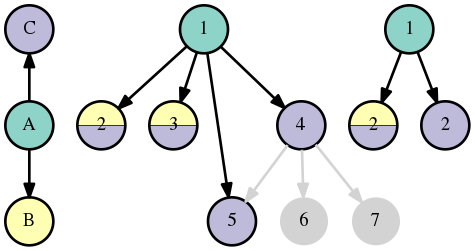}
    \caption{The template (left) and world (center) from \figref{fig:toy-example} recolored to represent solution $(B \to \{2, 3\}, A \to 1, C \to \{2, 3, 4, 5\})$. Each world node is colored with the same color as template nodes which can map to it or gray if no node maps to it. The right graph compresses the world graph by dropping nonparticipant nodes and combining nodes of the same color into a node with a label indicating the amount combined.}
    \label{fig:toy-example-compressed}
\end{figure}

\section{Compact Solution Representation}
\label{sec:compact-solution-representation}
\changedOne{As noted in prior sections, the solution space for subgraph matching is often combinatorially complex.  However, the notion of structural equivalence provides methods to diagram the solution space in a compact visual way that a user can understand.}
% The representative solutions found during subgraph search and the associated solution classes can produce a compact visual description of the solution space.
%In this section, we demonstrate how to create these expressive graphical representations.
We explain this first for the toy example from \figref{fig:toy-example}. We begin with the base representation of one solution, $(A\to 1, B\to2, C\to3)$, which pairs template vertices with world vertices and extend it to incorporate equivalence.

If we use template structural equivalence, equivalent template nodes are interchangeable. We indicate this by pairing template equivalence classes with world vertices; producing a solution amounts to picking a unique representative from each class. In our example, as B and C are equivalent, we write $(A\to1, \{B, C\} \to 2, \{B, C\}\to 3)$.
World equivalence is represented similarly: the representation $(A\to 1, B\to \{2,3\},C\to 4)$ indicates $B$ can match with 2 or 3.

Table \ref{tab:toy-representations} illustrates the different numbers of representative solutions for each different equivalence level for the toy problem in \figref{fig:toy-example}. The candidate and node cover equivalence numbers are computed assuming $A$ is assigned first. As the full candidate and node cover equivalence levels have the broadest notion of equivalence, they have the most compression.

If we use candidate or node cover equivalence, equivalence classes are recomputed before each assignment. Hence, it may be the case that a template node is paired with a world equivalence class that has been previously assigned but has grown in size due to recomputing equivalence.
For example, if we first assign template node B to the equivalence class $\{2, 3\}$, we are forced to assign $A$ to 1. Finally, we recompute equivalence, and we find that $\{2, 3, 4, 5\}$ comprise an equivalence class, to which we assign our last template node $C$. We therefore have solution class $(B \to \{2, 3\}, A \to 1, C \to \{2, 3, 4, 5\})$.
We diagram this class in \figref{fig:toy-example-compressed} where we color each template node and its associated candidates the same color. 
The subgraph of all nodes and edges that participate in the representative solution is the \textbf{solution-induced world subgraph}. 
The final graph depicts the \textbf{compressed solution-induced world subgraph} where we drop all nonparticipant nodes and edges, and we combine like-colored nodes into ``supernodes'' with a label indicating the number of nodes joined. 
From this last graph, we can observe the original template graph structure among the participant world nodes.

\begin{table}[]
    \centering
        \caption{Number of representative solutions for each equivalence level in the toy problem in \figref{fig:toy-example}}
    \label{tab:toy-representations}
    \begin{tabular}{c|c|c}
    \toprule
    Eq. Level & \# Rep. Sols. & Example Sol. \\ \midrule
    NE &   18 & $A\to1, B\to2,C\to3$\\
    TE & 9 & $A\to1, \{B,C\}\to 2, \{B,C\}\to 3$ \\
    WE & 10 & $A\to 1, B\to \{2,3\}, C\to 4$ \\
    TEWE & 6 & $A\to1, \{B,C\}\to 5, \{B,C\}\to \{6,7\}$ \\
    CE & 5 & $A \to 1, B \to 2, C \to \{3, 4, 5\}$ \\
    FE & 2 & $A\to1, B\to \{2,3,4,5\}, C\to \{2,3,4,5\}$\\
    NC & 2 & $A\to1, B\to \{2,3,4,5\}, C\to \{2,3,4,5\}$\\ \bottomrule
    \end{tabular}

\end{table}

We use these graphical representations depict various symmetric features of our datasets.
As an example, we plot the template and world subgraph for an example from the \textit{biochemical} dataset in \figref{fig:biochemical}. 
From this depiction, we observe that there are multiple different sources from which equivalence may arise. 
One aspect is the large number of pairs of structurally equivalent nodes that are colored the same in the template graph.
A second source is leaf nodes on the template graph that can be mapped to large equivalence classes in the world graph.
By using an equivalence-informed subgraph search, we can expose exactly where these complexities arise. 
The compressed solution-induced world graph is depicted in \figref{fig:biochemical-compressed} and clearly shows the role each world node plays with respect to the template graph in a solution. 

\begin{figure*}
    \centering
    \includegraphics[width=0.49\linewidth]{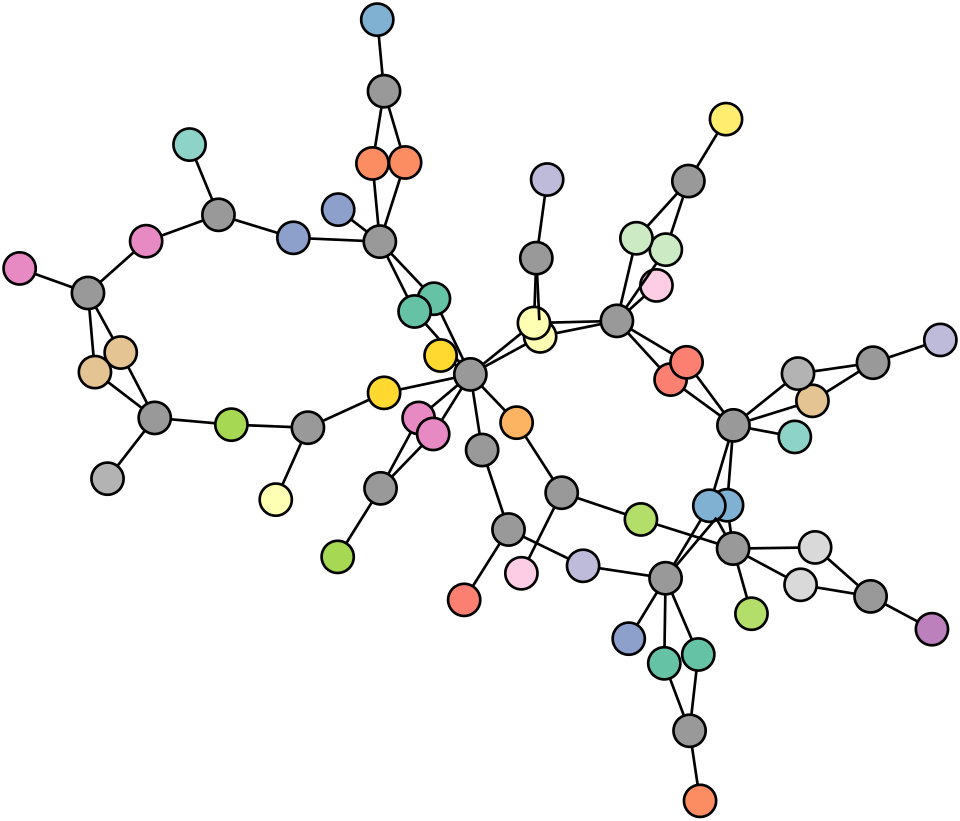}
    \includegraphics[width=0.49\linewidth]{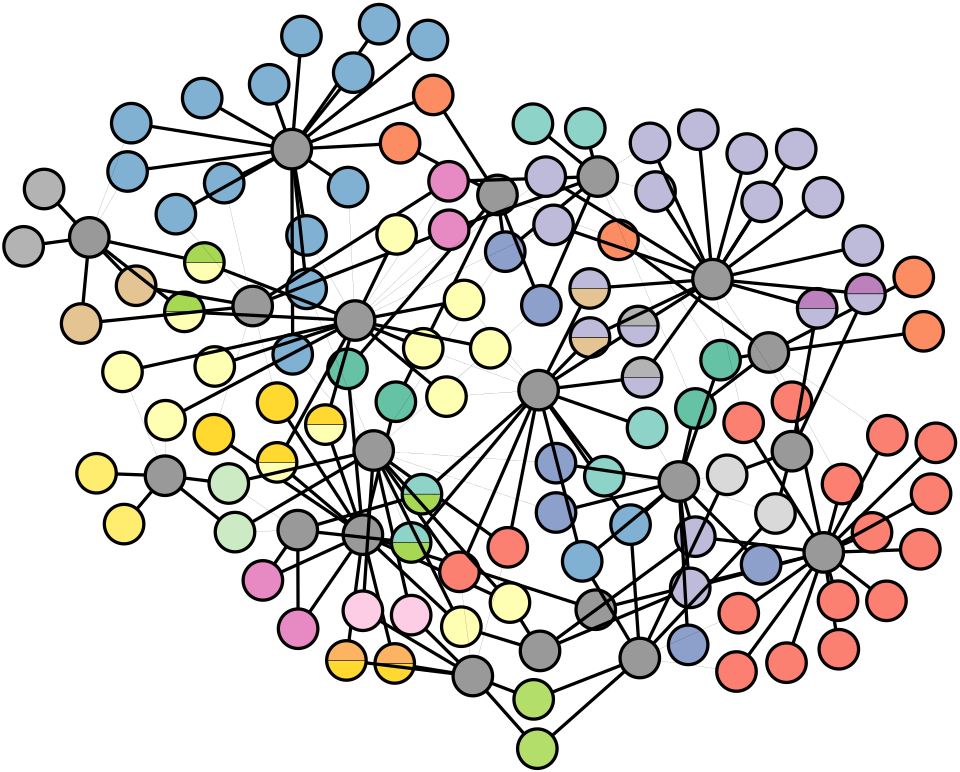}
    \caption{A biochemical reactions \cite{gay2014subgraph} template graph (left) and the solution-induced world subgraph (right) for a solution class comprised of $9.18\times 10^{13}$ solutions. Dark gray nodes are nodes with a single candidate. Nodes with the same non-gray color in the world subgraph are fully candidate equivalent. Nodes with two or more colors were part of one class at an early stage of subgraph search which was later merged into another class. All solutions represented by the compressed solution can be generated by mapping templates nodes of one color to world nodes with the same color.}
    \label{fig:biochemical}
\end{figure*}

\begin{figure}
    \centering
    \includegraphics[width=\linewidth]{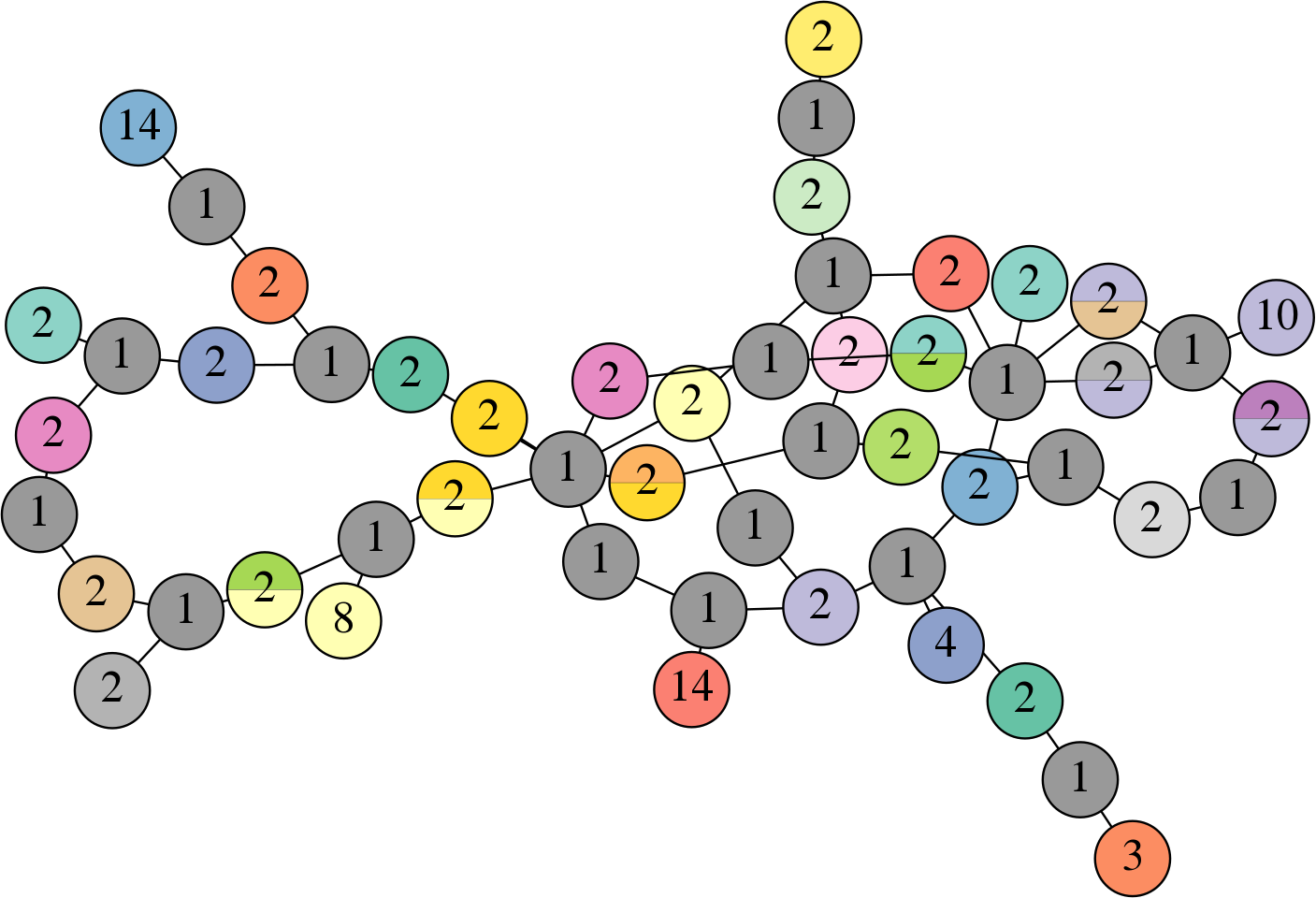}
    \caption{The world graph from \figref{fig:biochemical} with equivalent nodes joined into supernodes with numbers indicating the size of the class.}
    \label{fig:biochemical-compressed}
\end{figure}

% \figref{fig:mesh-example} shows a template and compressed solution-induced world subgraph for a representative solution from the meshes\_cv dataset. This particular example how an otherwise highly specified template with minimal symmetry can observe a combinatorial explosion of solutions. Even without any structural equivalence in the template or world, this solution can generate $4.4\times 10^{31}$ as there are a few leaf nodes in the template which each can be mapped to the very large group of nodes adjacent to the central node. 

% \begin{figure*}
%     \centering
%     \includegraphics[width=0.49\linewidth]{pictures/mesh_template.png}
%     \includegraphics[width=0.49\linewidth]{pictures/mesh_compressed_world.png}
%     \caption{Template and compressed solution-induced world subgraph for one representative solution from the meshes\_cv dataset representing approximately $4.4 \times 10^{31}$ solutions.}
%     \label{fig:mesh-example}
% \end{figure*}

\section{Application to Multiplex  Networks}
\label{sec:multiplex}

\subsection{Multiplex MultiGraph Matching}
Often analysts wish to encode attributed information into the nodes and edges of a graph and allow for more than one interaction to occur between nodes. For example, a transportation network may have multiple modes of travel between hubs (e.g., trains and subways).
Formally, if we have $K$ distinct edge labels, then a \textbf{multiplex multigraph} is a $K+1$-tuple $(V, E^1, E^2, \ldots, E^K)$ where $V$ is the set of vertices, and $E^i : V\times V \to \mathbb{Z}_{\ge 0}$ is a function dictating how many edges there are of label $i$ between two vertices. Intuitively, a multiplex multigraph is a collection of $K$ multigraphs which share the same set of nodes. The index $i$ of the edge function is the ``channel'' $i$, and we refer to the edges given by edge function $E^i$ as the edges in channel $i$, and the graph $(V, E^i)$ as the graph in channel $i$.

A \textbf{multiplex subgraph isomorphism} $f:V_T \to V_W$ preserves the number of edges in each channel. Given template $(V_T, E^1_T,\ldots,E^K_T)$ and world $(V_W, E^1_W,\ldots,E^K_W)$, for any $u, v \in V_T$, we require $E^i_W(f(u), f(v)) \ge E^i_T(u, v)$, i.e., there need to be enough edges between $f(u)$ and $f(v)$ to support the edges between $u$ and $v$. Definitions for equivalence also extend naturally: we say $v \sim_s w$ if in each channel $i$ for each $u \ne v, w$, $E^i(v, u) = E^i(w, u)$, $E^i(u, v) = E^i(u, w)$, and $E^i(v, w) = E^i(w, v)$. The other forms of equivalence and related theorems all generalize similarly.

Recently, significant work has been done on developing algorithms for finding multiplex subgraph isomorphisms. \cite{ingalalli2016sumgra} develops an indexing approach based on neighborhood structure in multichannel graphs. \cite{micale2020multiri} extends the single channel package \cite{bonnici2013subgraph} to handle the multichannel case and focuses on using intelligent vertex ordering for finding isomorphisms. \cite{moorman2018filtering} utilizes a constraint programming approach for filtering out candidates which is extended in \cite{moorman2021subgraph}. A similar filtering approach is taken in \cite{liu2019g}. \cite{MatchedFilter2019} relaxes the problem to a continuous optimization problem which is then solved and projected back onto the original space.

\subsection{Multiplex Experiments}

We assess the performance of our equivalence enhancements on the Glasgow solver, adapted to handle multiplex subgraph isomorphism problems. The adaptations involve minimal changes to the base algorithm, to ensure that matches are only made if they preserve the edges in every channel. To eliminate more candidates, we also perform a prefilter using the statistics and topology filters from \cite{moorman2018filtering} as well as maintain the subgraphs in each channel as the supplemental graphs used in the Glasgow algorithm.

We consider datasets including those from \cite{moorman2021subgraph} and which represent both real world examples and synthetically generated data. The real world examples include a transportation network in Great Britain \cite{fan2012graph}, an airline network \cite{cardillo2013emergence}, a social network built on interactions on Twitter related to the Higgs Boson \cite{de2013anatomy}, and COVID data \cite{zucker2021leveraging}. \changedOne{For the transportation and twitter networks, the template is extracted from the world graph.}
The synthetically generated datasets are examples which represent emails, phone calls, financial transactions, among other interactions between individuals \changedOne{and are all generated as part of the DARPA-MAA program \cite{GORDIAN,IVYSYS,PNNL}}. 
The subgraph isomorphisms to be detected may be a group of actors involved in adversarial activities including human trafficking and money laundering. The statistics regarding these different subgraph isomorphism problems are described in Table \ref{tab:multiplex-graphs}. \changedOne{For more details on these particular datasets, see \cite{moorman2021subgraph}.}

\begin{table}[]
    \caption{Data on Multiplex Graphs}
    \label{tab:multiplex-graphs}
    \centering
    \begin{tabular}{cccccc}
    \toprule
            & \multicolumn{2}{c}{Template} & \multicolumn{2}{c}{World} & \\
    Dataset & Nodes & Edges & Nodes & Edges & Chan. \\ \midrule
    Brit. Trans. & 53 & 56 & 262377 & 475502 & 5 \\
    Higgs Twitter & 115 & 2668 & 456626 & 5367315 & 4 \\
    Airlines & 37 & 210 & 450 & 7177 & 37\\ \midrule
    PNNL RW & 74 & 35 & 158 & 6407 & 3 \\
    PNNL v6-b0-s0 & 74 & 1620 & 22996 & 12318861 & 7 \\
    PNNL v6-b5-s0 & 64 & 1201 & 22994 & 12324975 & 7 \\
    PNNL v6-b1-s1 & 75 & 1335 & 22982 & 12324340 & 7 \\
    PNNL v6-b7-s1 & 81 & 1373 & 23011 & 12327168 & 7 \\
 \midrule 
    GORDIAN v7-1 & 156 & 3045 & 190869 & 123267100 & 10 \\
    GORDIAN v7-2 & 92 & 715 & 190869 & 123264754 & 10 \\ \midrule
    IvySys v7 & 92 & 195 & 2488 & 5470970 & 3 \\
    IvySys v11 & 103 & 387 & 1404 & 5719030 & 5 \\ \midrule
    COVID & 28 & 38 & 87580 & 1736985 & 9 \\ \midrule 
    Twitter - ER & 5-15 & 4-31 & 456626 & 5367315 & 4 \\
    \bottomrule
    \end{tabular}
\end{table}

The multiplex datasets are much larger than the single-channel graphs in the previous section, with the largest world graphs having hundreds of thousands of nodes and hundreds of millions of edges. The synthetic datasets are divided into three groups based on which organization generated the dataset: PNNL \cite{PNNL}, GORDIAN \cite{GORDIAN}, and IvySys Technologies \cite{IVYSYS}.

For our experiments, we examine the same seven modes of equivalence used in the single channel case, but with a time limit of one hour to count as many solutions as possible. These experiments were run on the same computer using our adapted version of the Glasgow solver. The amount of time required to enumerate all the solutions is displayed in Table \ref{tab:multichannel-times} and the number of solutions found with a given method is displayed in Table \ref{tab:multichannel-solutions}.
\begin{table*}[]
    \centering
    \caption{Time (s) to enumerate solution spaces of multichannel problems. Experiments timed out at one hour. The Twitter-ER dataset is averaged over a collection of problems and timed out after ten minutes.}
    \label{tab:multichannel-times}
\begin{tabular}{lrrrrrrr}
\toprule
Algorithm &       CE &       FE &       NC &       NE &       TE &     TEWE &       WE \\
Dataset              &          &          &          &          &          &          &          \\
\midrule
Brit. Trans.    &  3600 &  3600 &  3600 &  3600 &  3600 &  3600 &  3600 \\
Higgs Twitter           &  3600 &   \textbf{369} &   456 &  3600 &  3600 &  3600 &  3600 \\
Airlines           &     0.34 &     \textbf{0.24} &  1985 &  3600 &  1329 &  3600 &  3600 \\ \midrule
PNNL RW           &  3600 &  3600 &  3600 &  3600 &  3600 &  3600 &  3600 \\
PNNL v6-b0-s0     &    41.6 &    42.1 &    41.8 &    \textbf{41.3} &    41.7 &    41.4 &    41.6 \\
PNNL v6-b1-s1     &   241 &   \textbf{240} &   241 &   \textbf{240} &   242 &   \textbf{240} &   \textbf{240} \\
PNNL v6-b5-s0     &    62.6 &    58.1 &    58.9 &    \textbf{58.0} &    62.3 &    62.3 &    63.0 \\
PNNL v6-b7-s1     &  1133 &   200 &   201 &  3600 &   \textbf{188} &   211 &  1138 \\ \midrule
GORDIAN v7-1 &  3600 &   \textbf{327} &  3600 &  3600 &  3600 &  3600 &  3600 \\
GORDIAN v7-2 &  3600 &   \textbf{316} &  3600 &  3600 &  3600 &  3600 &  3600 \\ \midrule
IvySys v7         &  3600 &  3600 &  3600 &  3600 &  3600 &  3600 &  3600 \\
IvySys v11        &  3600 &  3600 &  3600 &  3600 &  3600 &  3600 &  3600 \\ \midrule
COVID & 3600 & 3600 & 3600 & 3600 & 3600 & 3600 & 3600 \\ \midrule
Twitter - ER & 514.7 & 505.4 & \textbf{494.8} & 536.0 & 536.7 & 544.2 & 541.3 \\
\bottomrule
\end{tabular}

\end{table*}
\begin{table*}
\centering
\caption{Number of solutions found for multichannel problems within one hour. The Twitter-ER dataset is averaged over a collection of problems and is timed out after ten minutes.}
\label{tab:multichannel-solutions}
\begin{tabular}{lrrrrrrr}
\toprule
Algorithm &       CE &        FE &        NC &        NE &        TE &      TEWE &        WE \\
Dataset              &           &           &           &           &           &           &           \\
\midrule
Brit. Trans.    &  1.48e+11 &  \textbf{2.34e+15} &  4.97e+08 &  1.27e+07 &  2.48e+12 &  2.02e+12 &  1.17e+07 \\
Higgs Twitter           &  1.38e+14 &  \textbf{3.23e+14} &  \textbf{3.23e+14} &  5.65e+06 &  6.44e+06 &  5.72e+06 &  6.95e+06 \\
Airlines           &  \textbf{3.67e+09} &  \textbf{3.67e+09} &  \textbf{3.67e+09} &  3.55e+09 &  3.65e+09 &  8.11e+08 &  2.35e+09 \\ \midrule
PNNL RW           &  3.50e+09 &  2.78e+11 &  \textbf{4.72e+11} &  8.59e+08 &  2.01e+10 &  5.00e+09 &  8.70e+08 \\
PNNL v6-b0-s0     &  \textbf{1.15e+03} &  \textbf{1.15e+03} &  \textbf{1.15e+03} &  \textbf{1.15e+03} &  \textbf{1.15e+03} &  \textbf{1.15e+03} &  \textbf{1.15e+03} \\
PNNL v6-b1-s1     &  \textbf{1.15e+03} &  \textbf{1.15e+03} &  \textbf{1.15e+03} &  \textbf{1.15e+03} &  \textbf{1.15e+03} &  \textbf{1.15e+03} &  \textbf{1.15e+03} \\
PNNL v6-b5-s0     &  \textbf{1.15e+03} &  \textbf{1.15e+03} &  \textbf{1.15e+03} &  \textbf{1.15e+03} &  \textbf{1.15e+03} &  \textbf{1.15e+03} &  \textbf{1.15e+03} \\
PNNL v6-b7-s1     &  \textbf{3.14e+08} &  \textbf{3.14e+08} &  \textbf{3.14e+08} &  8.57e+07 &  \textbf{3.14e+08} &  \textbf{3.14e+08} &  \textbf{3.14e+08} \\ \midrule
GORDIAN v7-1 &  1.11e+12 &  \textbf{9.13e+12} &  1.35e+10 &  1.61e+07 &  7.84e+09 &  5.34e+09 &  1.58e+07 \\
GORDIAN v7-2 &  2.14e+11 &  \textbf{1.35e+16} &  1.15e+15 &  1.72e+07 &  3.88e+08 &  3.19e+08 &  1.65e+07 \\ \midrule
IvySys v7         &  2.04e+14 &  \textbf{8.04e+96} &  2.09e+90 &  1.75e+09 &  2.67e+47 &  5.84e+45 &  5.39e+07 \\
IvySys v11        &  7.41e+10 &  \textbf{3.64e+89} &  5.09e+66 &  1.77e+09 &  4.43e+72 &  6.64e+71 &  4.88e+07 \\ \midrule
COVID & 5.27e+14 & \textbf{7.45e+21} & 3.11e+20 & 9.63e+06 & 9.53e+06 & 4.51e+07 & 1.31e+07 \\ \midrule
Twitter - ER & 3.52e+08 & 8.84e+09 & \textbf{1.40e+11} & 7.17e+05 & 7.41e+05 & 6.14e+05 & 6.78 e+05 \\
\bottomrule
\end{tabular}
\end{table*}
A quick inspection of the times illustrates that in a few cases (Airlines, GORDIAN, and Higgs Twitter), using full equivalence can enumerate the full solution space an order of magnitude faster than any other approach. This speedup is reflected in the solution count table for which FE finds significantly many more solutions. The other methods only find a mere fraction of the total solutions. The NC method often appears to be the second best both in terms of solutions found and time taken to enumerate all. This makes sense given Proposition \ref{prop:equiv-hierarchy}. TE appears to be the third best method which can be explained by the simplicity of implementation and having no need to recompute equivalence. WE and CE are not competitive with the other methods.
The datasets bear different qualities that illustrate why certain levels of equivalence work better than others. We discuss a few datasets in detail.

\subsubsection{PNNL}
\changedOne{The PNNL template and world graphs \cite{PNNL} are generated to model specific communication, travel, and transaction patterns from real data and the templates are then embedded into the world graph}. For these instances, the counting problem is almost entirely solved after applying the initial filter and the solution space is understood by equivalence in the template. 
\begin{figure}
    \centering
    \includegraphics[width=0.9\linewidth]{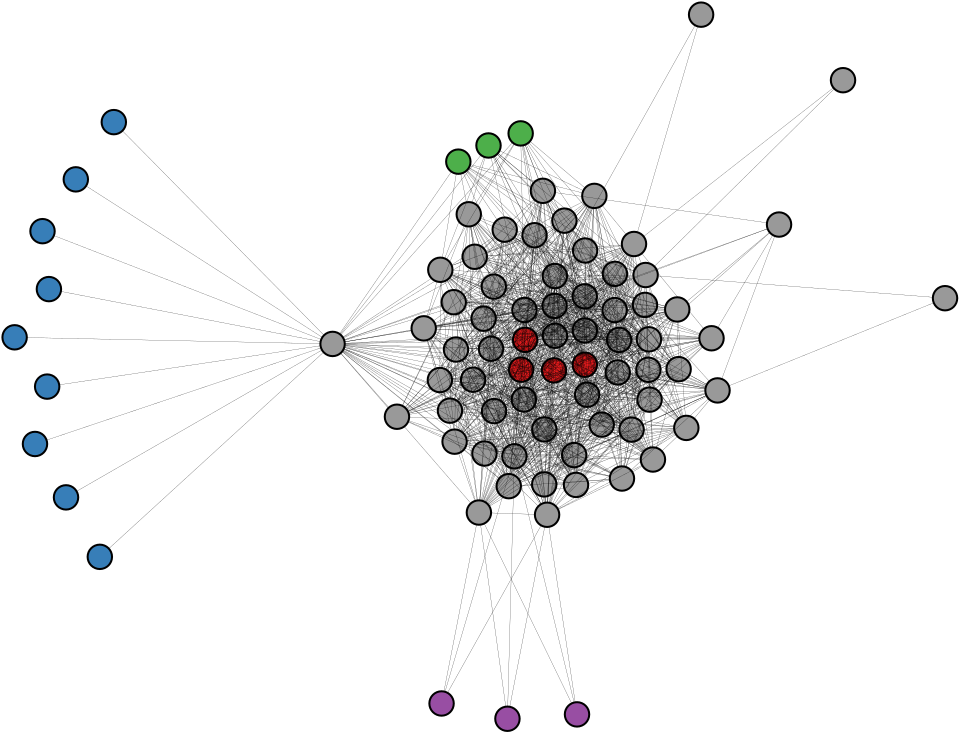}
    \caption{Template Graph for PNNL v6-b7-s1. Non-gray nodes of the same color are structurally equivalent.}
    \label{fig:pnnl-v6-b7-s1}
\end{figure}
For example, observe the template displayed in \figref{fig:pnnl-v6-b7-s1}.
The number of solutions generated equals the count of solutions generated by permutations of the template nodes for a single representative solution. We have a group of 9, a group of 4, and two groups of 3 interchangeable nodes, meaning any solution can generate $9!4!3!3!$ more solutions. All variants on the PNNL problems illustrate this behavior. 

\subsubsection{GORDIAN}

\changedOne{The GORDIAN datasets \cite{GORDIAN} have a much larger templates and worlds than PNNL and they are generated separately in an agent-based fashion to match the daily routines and travel patterns of a certain population of people}. Only the FE method fully enumerates the solution space, but the NC and CE methods come close to a full enumeration.
\begin{figure*}
    \centering
    \includegraphics[width=0.3\linewidth]{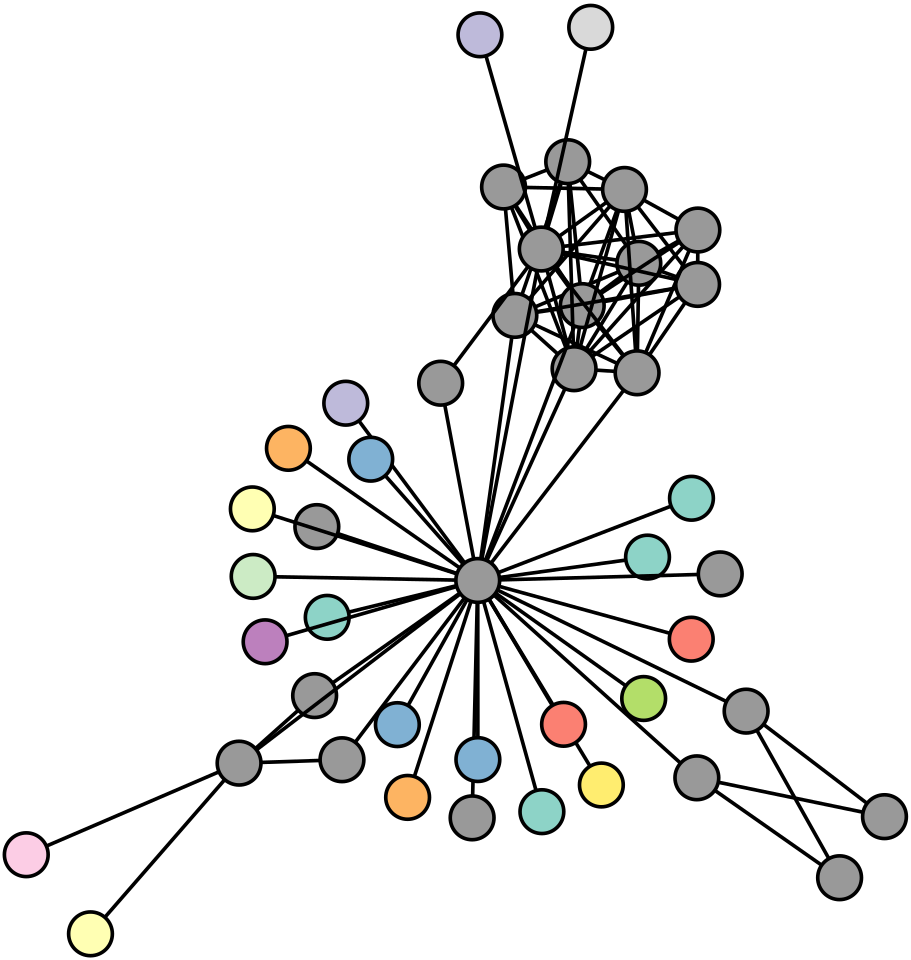}
    \includegraphics[width=0.3\linewidth]{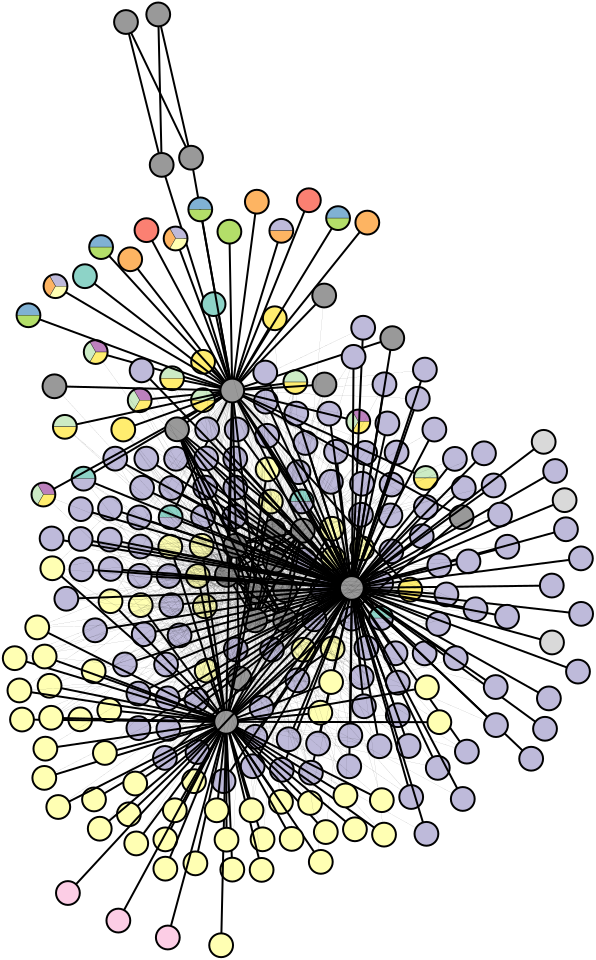}
    \includegraphics[width=0.35\linewidth]{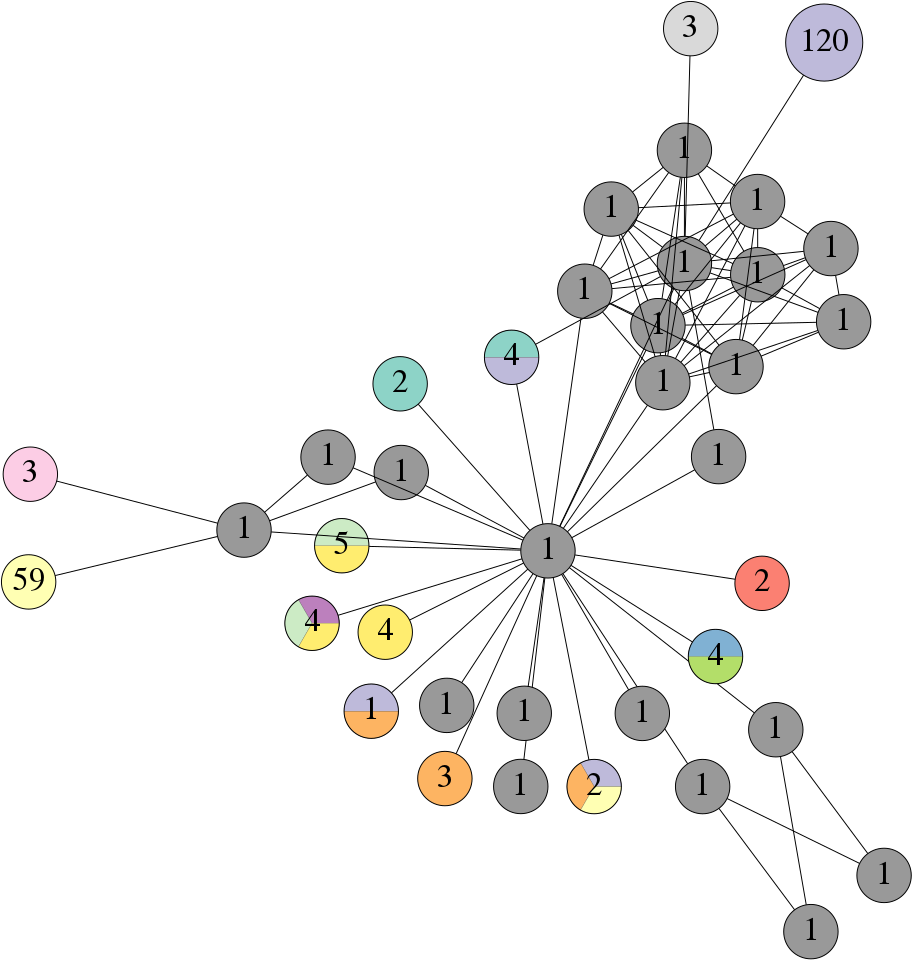}
    \caption{Template (left), solution-induced world subgraph (center) and compressed solution-induced world subgraph (right) for a solution class which can generate about $3 \times 10^{12}$ solutions to GORDIAN v7-2 \cite{GORDIAN}. World nodes of the same color are fully candidate equivalent and are candidates of the template node of the same color. All solutions represented by this compressed solution can be generated by mapping each colored node to one of groups of world nodes with the same color.}
    \label{fig:gordian-v2}
\end{figure*}
\figref{fig:gordian-v2} illustrates the symmetries for one solution class; the template graph possesses a large group of leaf nodes. After mapping the central node to a candidate, the leaves need only be mapped to neighbors of this candidate. These graphs demonstrate a trade-off between node specificity and symmetry: template nodes with fewer edges exhibit great amounts of symmetry, whereas dense subgraphs are restricted in their candidates and have minimal symmetry. The right graph in \figref{fig:gordian-v2} depicts a compressed version of the world graph induced by this solution class from which $3\times 10^{12}$ solutions may be generated.

\begin{figure*}
    \centering
    \includegraphics[width=0.45\linewidth]{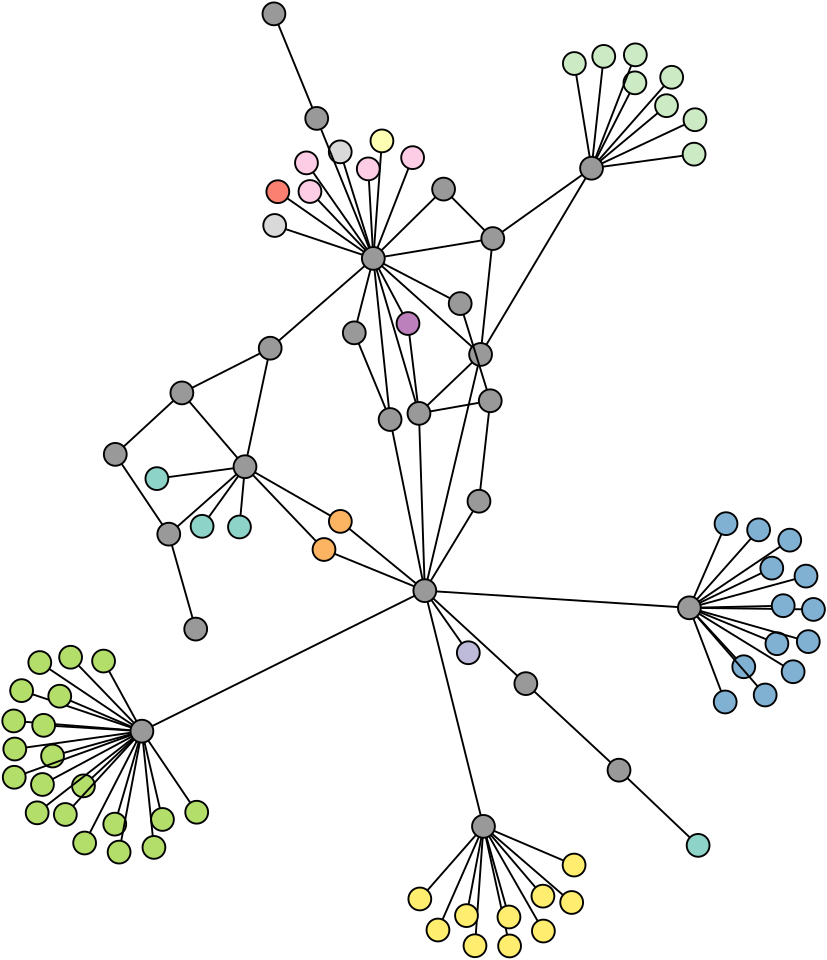}
    \includegraphics[width=0.45\linewidth]{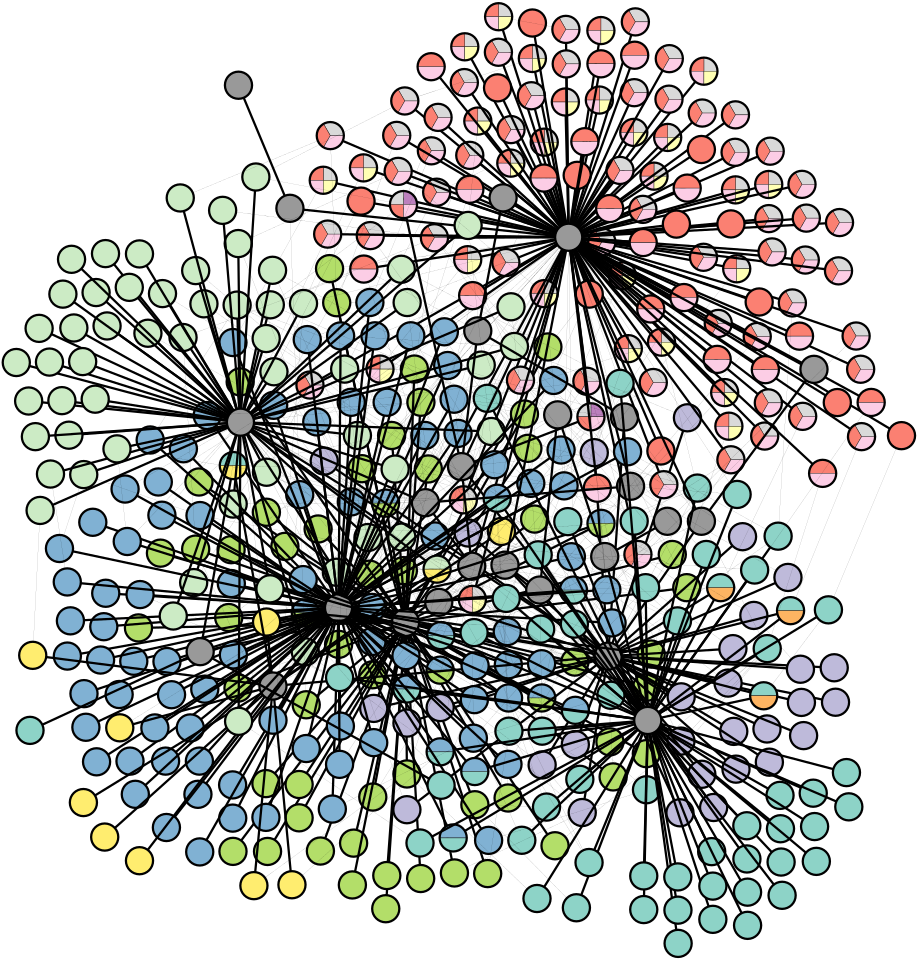}
    \caption{Template (left) and solution-induced world subgraph (right) for a solution class from which $7.82\times 10^{103}$ solutions to IvySys v7 \cite{IVYSYS} can be generated. World nodes of the same color are fully candidate equivalent and are candidates of the template node of the same color. All solutions represented by this compressed solution can be generated by mapping each colored node to one of groups of world nodes with the same color.}
    \label{fig:ivysys-v7}
\end{figure*}

\subsubsection{IvySys}

\changedOne{The Ivysys template and world graphs \cite{IVYSYS} are separately generated to match the degree distribution and email behavior of the Enron email dataset and have the most complex solution space.} None of the methods were successful at enumerating all solutions. The vastness of the solution space is in contrast to the size of the graphs which only have thousands of nodes. The complexity emerges from the preponderance of template leaf nodes as shown in \figref{fig:ivysys-v7}, depicting one solution class from which $7.82\times 10^{103}$ solutions may be generated. \figref{fig:ivysys-v7-venn} depicts the compressed representation of the world subgraph for this solution as well as a Venn diagram displaying candidates of certain template nodes.

%This template-level equivalence contributes to the combinatorial explosion of solutions as t
The TE solver finds an astonishing $10^{47}$ solutions for IvySys v7. However, using the FE method still dramatically increases the solution count, by %enabling the assignment of 
mapping these large template equivalent classes into larger world equivalence classes. 
An equivalence-informed subgraph search is essential as the NE method finds only $1.75 \times 10^9$ solutions, 90 orders of magnitude less than the FE search. Furthermore, a typical subgraph search would assign each group of leaf nodes sequentially meaning only the candidates of the last group would be explored. Incorporating symmetry gives a fuller vision of the solution space.

\begin{figure*}
    \centering
    \includegraphics[width=0.8\linewidth]{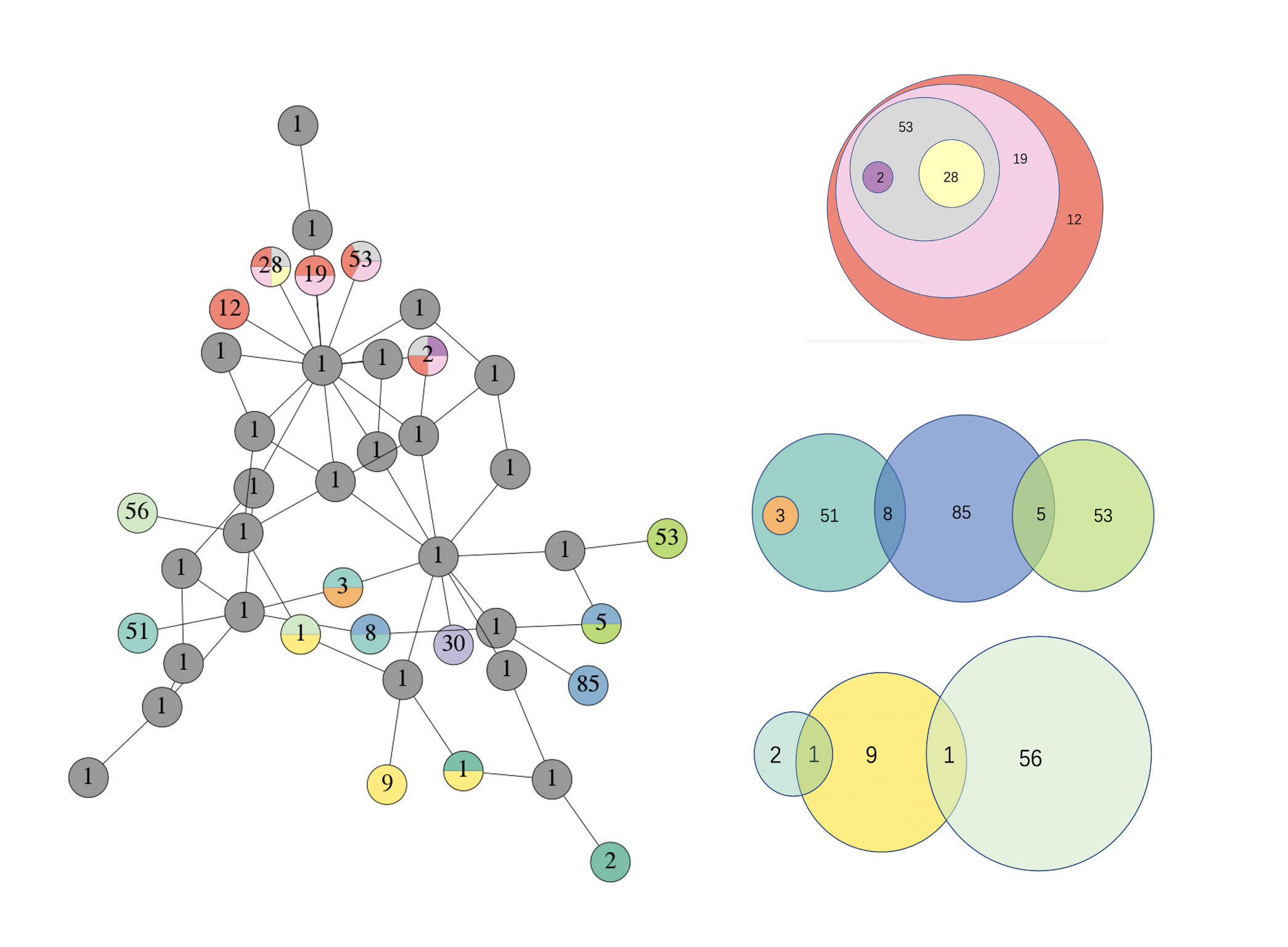}
     \caption{IvySys v7 \cite{IVYSYS} Compressed solution-induced world graph (left) and the Venn diagram representation of intersecting candidate sets in world graph(right) for a solution class from which $7.82\times 10^{103}$ solutions to  can be generated. The number in each section in the Venn diagram represents the size of a node cover equivalence class in the world graph. All solutions represented by this compressed solution can be generated by mapping each colored node in the template to the set in the Venn diagram with the same color.}
    \label{fig:ivysys-v7-venn}
\end{figure*}
% \begin{figure}
%     \centering
%     \includegraphics[width=\linewidth]{pictures/ivysys_v7_world_compressed.png}
%     \caption{The world subgraph from \figref{fig:ivysys-v7} with nodes with same color merged into supernodes with a label indicating the number combined.}
%     \label{fig:ivysys-v7-compressed}
% \end{figure}

\subsubsection{COVID}
We lastly apply our algorithm to the problem of querying a knowledge graph representing known causal relations between a large variety of biochemical entities. This problem arises from a desire to extracting causal knowledge in an automated fashion from the research literature. In \cite{zucker2021leveraging}, a knowledge graph is assembled from multiple sources including the COVID-19 Open Research Dataset \cite{wang2020cord}, the Blender Knowledge Graph \cite{wang2020covid}, and the comparative toxigenomics database \cite{davis2021ctd}. The authors of \cite{zucker2021leveraging} then create a query representing how SARS-CoV-2 might cause a pathway leading to a cytokine-storm in COVID-19 patients, but is generalized to detect other possible confounding factors in the pathway.

When rephrased as a multichannel subgraph isomorphism problem, template and world nodes represent biochemical entities. 
Some template nodes are specified, and others are labeled as a chemical, gene or protein. 
The 9 channels in this problem are various known types of interactions between entities, e.g., activation. A solution is an assignment of each node which has the desired chemical interactions.

As can be seen in Tables \ref{tab:multichannel-times} and \ref{tab:multichannel-solutions}, there is an abundance of solutions to this problem, and incorporating equivalence greatly enhances our ability to understand the solution space.
\figref{fig:covid} depicts the template and Venn diagrams of candidates sets for one solution class and exposes unspecified template nodes with a large amount of candidates. Such information is useful to an analyst for determining confounding factors in a pathway and suggesting label information or interactions to add to better specify the entire solution space.

\subsubsection{Higgs Twitter Erdős–Rényi Experiments}
\changedOne{Lastly, we perform a similar experiment as we did for single channel graphs using small Erdős–Rényi graphs as our templates and our largest graph, the Higgs Twitter dataset, as our world graph. 
We generate a multichannel template graph by overlaying 4 different graphs corresponding to each channel each generated as an Erdős–Rényi graph with $p = \frac{\log n_t}{8n_t}$ where $n_t$ is the number of template nodes. 
This value $p$ is chosen so that the graph will be connected with high probability. We generate 45 connected graphs in this way for each of $n_t=5,7,9,11,13,15$. We then compute the number of isomorphisms counted for each method within 10 minutes. For these problems, we precompute the world structural equivalence classes of the Higgs Twitter graph prior to running our algorithms. 
The average isomorphism count for each equivalence method and template size is  depicted in Figure \ref{fig:twitter-average-iso-count}. 
The overall averages for total runtime and isomorphism count are included in Tables V and VI under the \textit{Twitter-ER} dataset.}

\changedOne{From these results, we observe that the NC, FE, and CE methods find significantly more solutions than the base routine whereas the other equivalence methods do not improve on the NE method. NC performs the best both in terms of the number of isomorphisms found and the total amount of time which we speculate is due to its lightweight computation and ability to capture most of the equivalence. FE and CE are a few orders of magnitude worse, and the remaining methods TE, WE, and TEWE fail to provide significant benefit over the base method and in fact does worse when involving world structural equivalence. That these methods do not improve much we can explain by the fact that nodes in multichannel Erdős–Rényi graphs are fairly unlikely to be structurally equivalent. All in all, these experiments demonstrate even when using randomly generated template graphs, significant improvements can be had in incorporating equivalence into the algorithm. However, certain modes of equivalence may be more appropriate for certain classes of graphs and some care must be taken to ensure that the level of equivalence chosen actually helps with solving the problem.}

\begin{figure}
    \centering
    \includegraphics[width=0.9\linewidth]{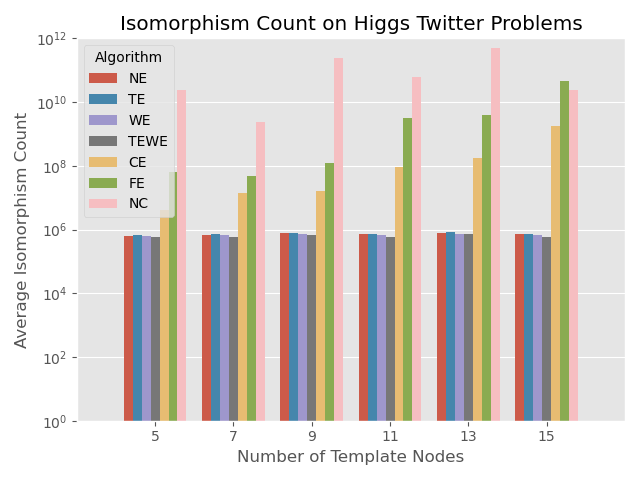}
    \caption{\changedOne{The average number of subgraph isomorphisms found for each equivalence level where the templates are small Erdős–Rényi graphs and the world is the \textit{Higgs Twitter} graph.}}
    \label{fig:twitter-average-iso-count}
\end{figure}

%\begin{figure*}
%    \centering
%    \includegraphics[width=0.8\linewidth]{pictures/covid_graphs.png}
%    \caption{Template (left) and compressed solution-induced world subgraph (right) from the COVID-19 data set from which $2.6 \times 10^{18}$ solutions can be generated. A few template nodes were specified at the start whereas others simply received a node label of C, P, or G indicating chemical, protein, and gene respectively. Edge colors indicate different types of interaction between nodes (e.g., increases expression of). The solutions may be generated by mapping non-gray template nodes of one color to world nodes of the same color.}
%    \label{fig:covid}
%\end{figure*}
\begin{figure*}
    \centering
    \includegraphics[width=0.8\linewidth]{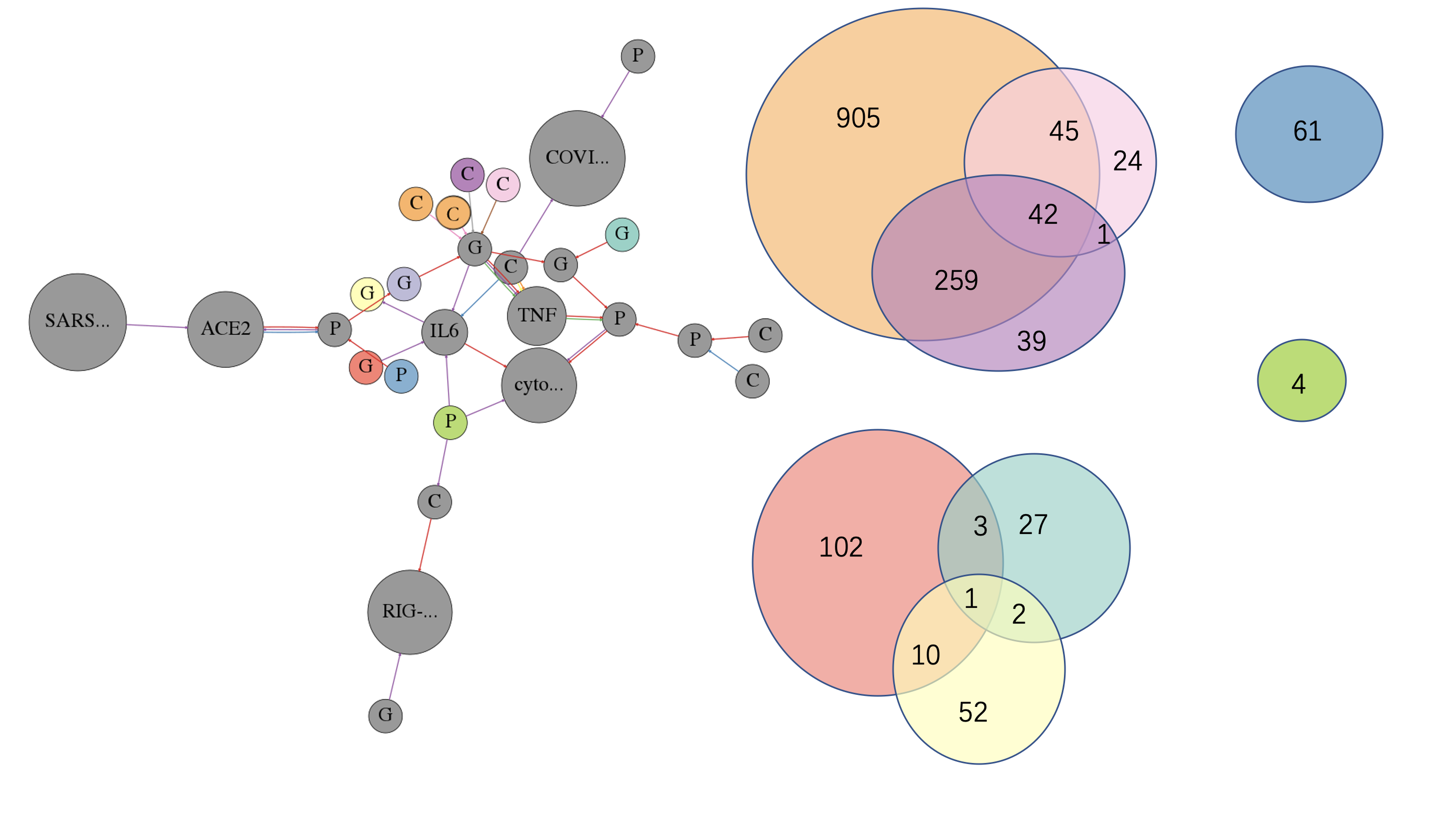}
    \caption{COVID-19 \cite{zucker2021leveraging} template (left) and the Venn diagram of candidate sets in world graph (right) from which $2.6 \times 10^{18}$ solutions can be generated in one solution class. Each section in the Venn Diagram represents a node cover equivalence class, and the number in the section is the size of the class. A few template nodes were specified at the start whereas others simply received a node label of C, P, or G indicating chemical, protein, and gene respectively. The solutions may be generated by mapping non-gray template nodes of one color to world nodes in the Venn diagram section of the same color.}
    \label{fig:covid}
\end{figure*}
\section{Conclusion}
In this work, we have developed a theory for static and dynamic notions of equivalence and presented conditions under which node assignments can be interchanged while preserving isomorphisms. 
With minimal changes to a subgraph isomorphism routine to incorporate equivalence during a tree search, we can dramatically reduce the amount of time to solve a problem and get a compact characterization of the solution space. For instances with minimal symmetry, little is to be gained, but for problems with large symmetric structures, it is essential to exploit equivalence in order to understand the large solution space. \changedOne{In particular, we demonstrated that the FE and NC methods both perform well in capturing equivalence present in the problem enabling the greatest compression of the solution space.} We showed our results apply to standard subgraph solvers by integrating our methods into the state-of-the-art solver Glasgow and extended our methods to the more complex problem spaces of multiplex multigraphs.

Future directions for this research include adapting these notions of equivalence to inexact search as well as producing inexact forms of equivalence. We would also like to better understand how to incorporate automorphic equivalence with the different notions of equivalence discussed in this paper. 

\section*{Acknowledgements}
This  material  is  based  on  research  sponsored  by  the Air Force Research Laboratory and DARPA under agreement number FA8750-18-2-0066. The U.S. Government is  authorized  to  reproduce  and  distribute  reprints  for Governmental  purposes  notwithstanding  any  copyright notation  thereon.  The  views  and  conclusions  contained herein  are  those  of  the  authors  and  should  not  be  interpreted as necessarily representing the official policies or  endorsements,  either  expressed  or  implied,  of  the Air Force Research Laboratory and DARPA or the U.S.Government.

This work was also supported by NSF Grant DMS-2027277.

\printbibliography
\begin{IEEEbiography}[{\includegraphics[width=1in,height=1.25in,clip,keepaspectratio]{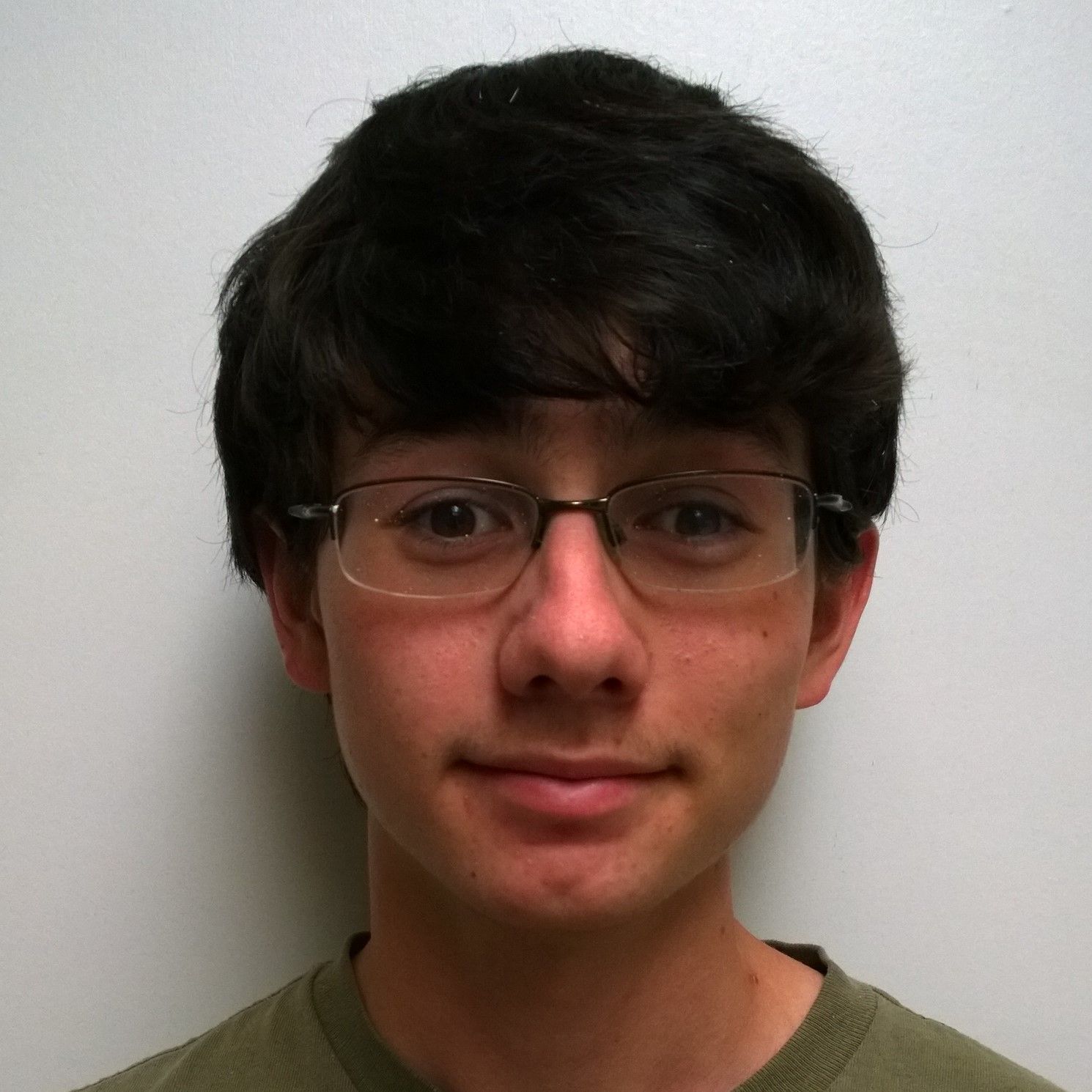}}]{Dominic Yang} is a Ph.D. candidate in the Department of Mathematics at the University of California, Los Angeles. He received B.S. degrees in Applied Mathematics and Computer Science from University of California, Davis in 2018. His research interests include graph algorithms, combinatorial optimization, mixed integer programming, and deep learning.
\end{IEEEbiography}

\begin{IEEEbiography}[{\includegraphics[width=1in,height=1.25in,clip,keepaspectratio]{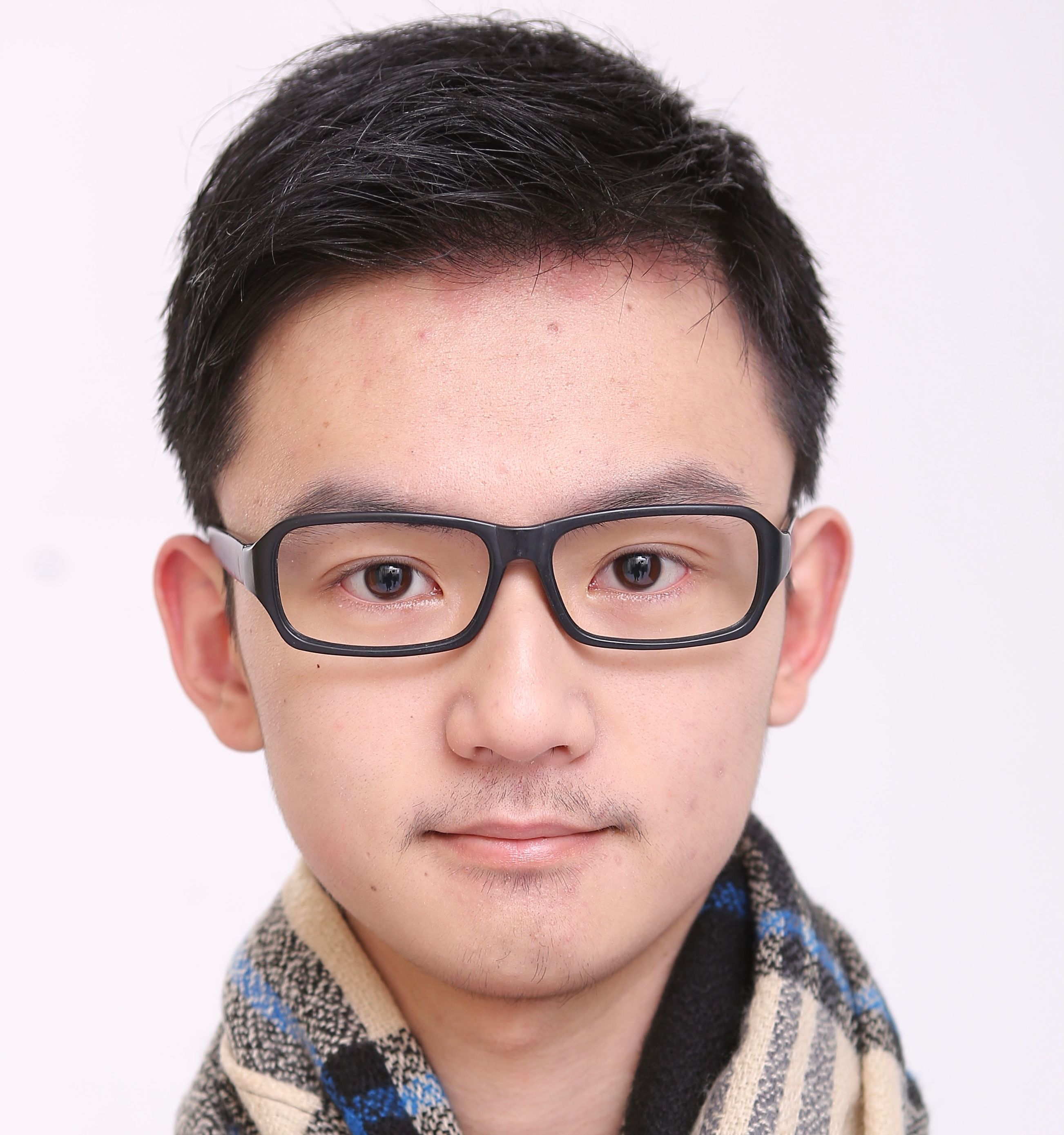}}]{Yurun Ge} is a Ph.D. candidate in the Department of Mathematics at the University of California, Los Angeles. He received his B.S. degree in Mathematics from Shanghai Jiaotong University in 2018. His research interest include graph algorithms and network analysis.
\end{IEEEbiography}

\begin{IEEEbiography}[{\includegraphics[width=1in,height=1.25in,clip,keepaspectratio]{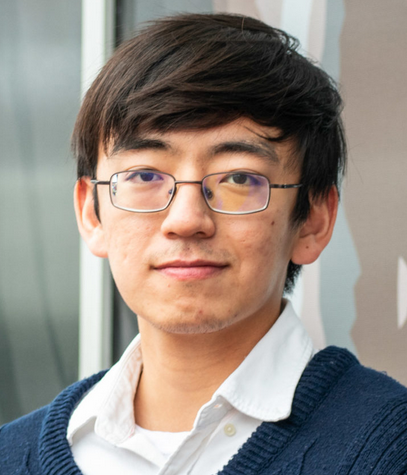}}]{Thien Nguyen} is a Ph.D. candidate in the Department of Computer Science at Northeastern University. He received his B.S. degree in Computer Science from University of California, Irvine in 2020. His research interests include  optimization and algorithms for robust learning. 
\end{IEEEbiography}

\begin{IEEEbiography}[{\includegraphics[width=1in,height=1.25in,clip,keepaspectratio]{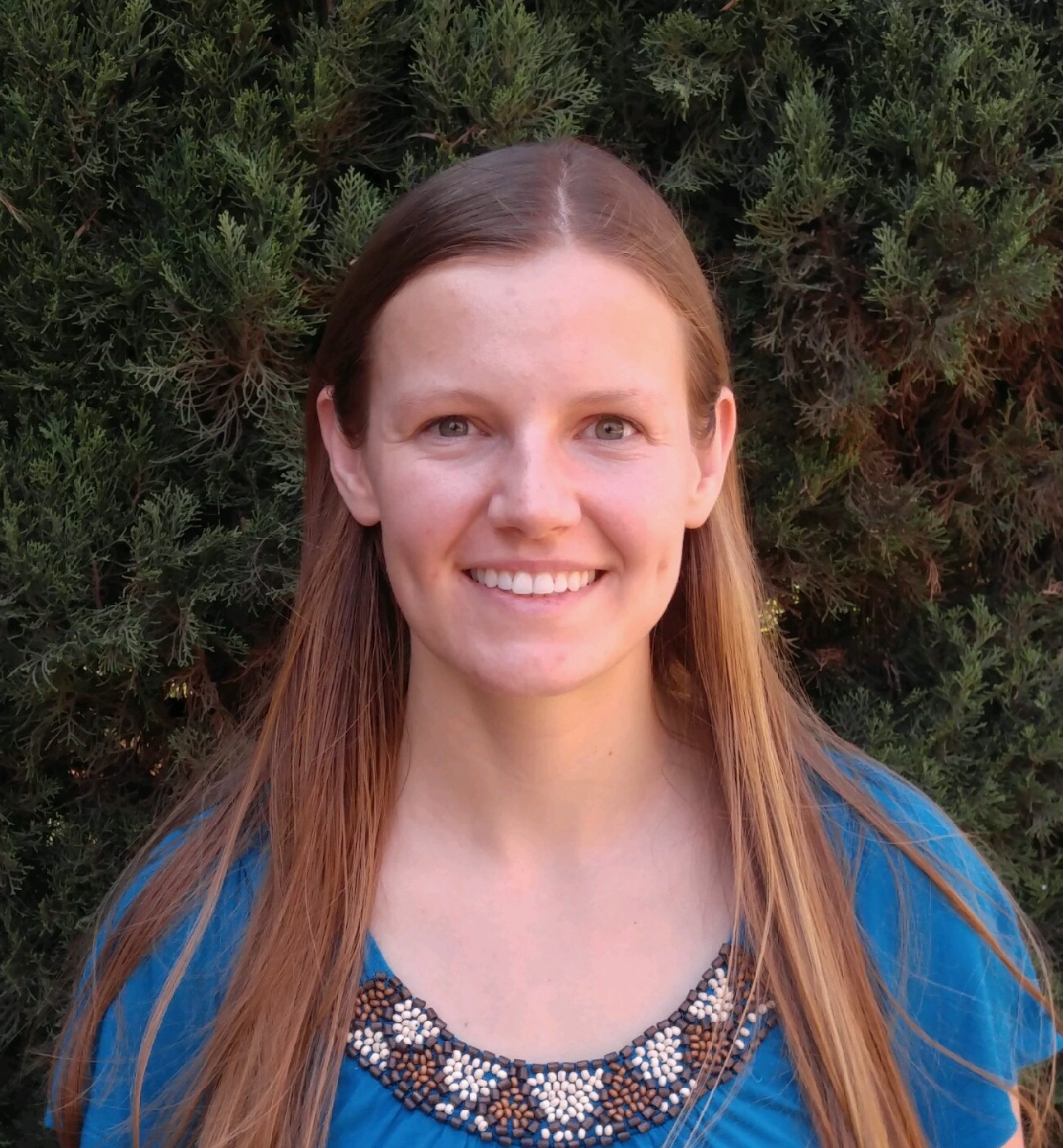}}]{Denali Molitor} completed her PhD in Mathematics at the University of California, Los Angeles in 2021. She received her B.A. degree in Mathematics from Colorado College in 2014. Her research interests include numerical linear algebra and stochastic optimization.
\end{IEEEbiography}

\begin{IEEEbiography}[{\includegraphics[width=1in,height=1.25in,clip,keepaspectratio]{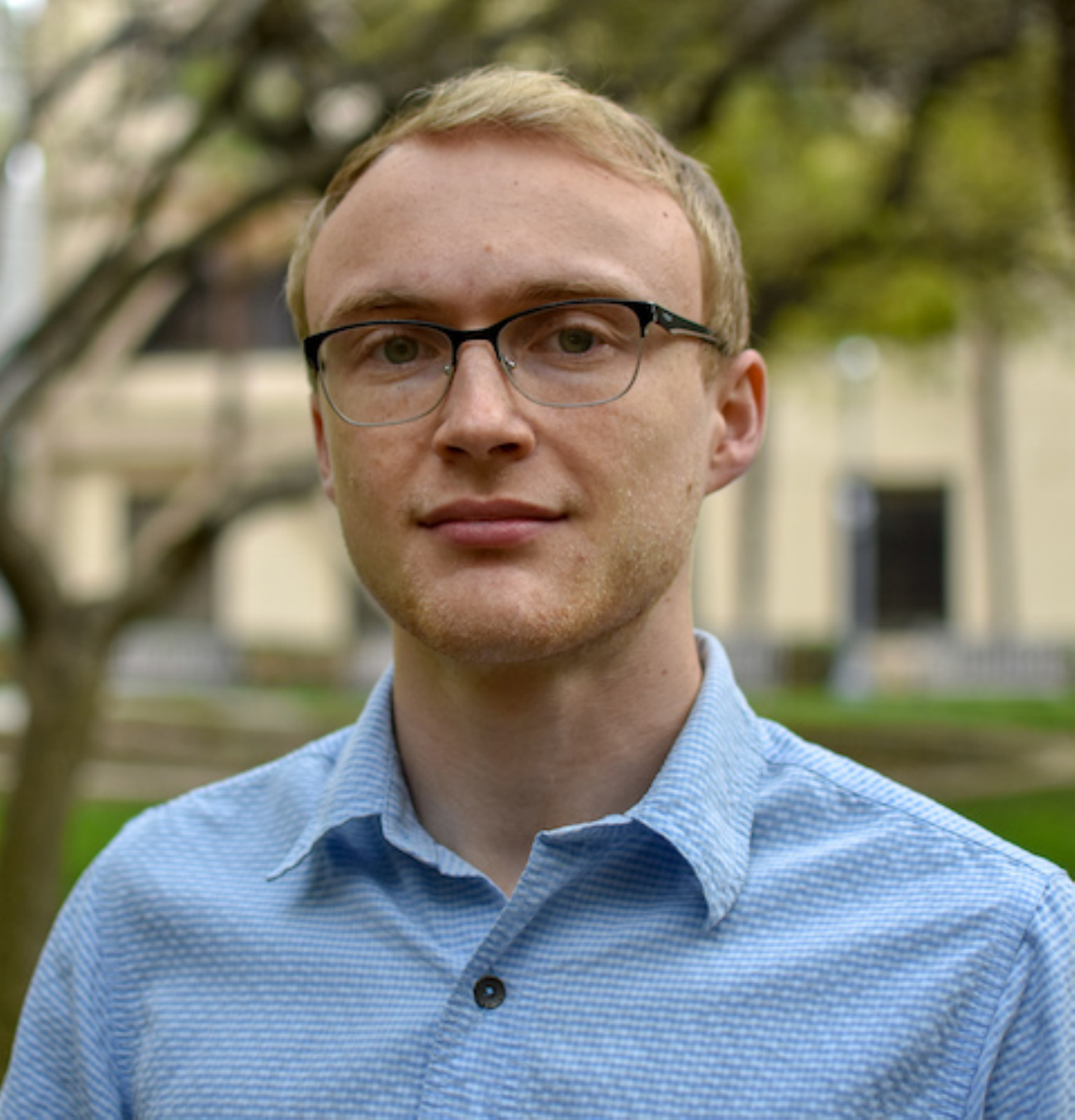}}]{Jacob D. Moorman} completed his PhD in Mathematics at the University of California, Los Angeles in 2021. He received B.S. degrees in Mathematics and Computer Science from the New Jersey Institute of Technology in 2016. His research interests include network analysis, linear algebra, stochastic optimization algorithms, pattern recognition, and high-dimensional data analysis.
\end{IEEEbiography}

\begin{IEEEbiography}[{\includegraphics[width=1in,height=1.25in,clip,keepaspectratio]{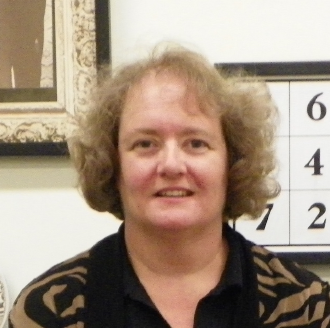}}]{Andrea L. Bertozzi, {\it Member, IEEE}} completed all her degrees in mathematics at Princeton. She is an Applied Mathematician with expertise in nonlinear partial differential equations and graphical models for machine learning. She was an L. E. Dickson Instructor and an NSF Postdoctoral Fellow at the University of Chicago from 1991 to 1995. She was the Maria Geoppert-Mayer Distinguished Scholar with the Argonne National Laboratory from 1995 to 1996. She was on the faculty at Duke University from 1995 to 2004, first as an Associate Professor of mathematics and then as a Professor of mathematics and physics. She has been on the faculty at UCLA since 2003 where she now holds the position of Distinguished Professor of Mathematics and Mechanical and Aerospace Engineering, along with the Betsy Wood Knapp Chair for Innovation and Creativity. She was elected to the American Academy of Arts and Sciences in 2010 and to the Fellows of the Society of Industrial and Applied Mathematics (SIAM) in 2010. She became a Fellow of the American Mathematical Society in 2013 and the American Physical Society in 2016. In 2018, she was elected to the U.S. National Academy of Sciences. Her honors include the Sloan Research Fellowship in 1995, the Presidential Early Career Award for Scientists and Engineers in 1996, and the SIAM Kovalevsky Prize in 2009, and the SIAM Kleinman Prize in 2019. She received the SIAM Outstanding Paper Prize in 2014 with A. Flenner, for her work on geometric graphbased algorithms for machine learning. She received the Simons Math + X Investigator Award in 2017.  %She is a Thomson-Reuters/Clarivate Analytics Highly Cited Researcher in mathematics in 2015 and 2016. 
\end{IEEEbiography}

\end{document}